\PassOptionsToPackage{prologue,usenames,dvipsnames,table}{xcolor}
\PassOptionsToPackage{bookmarks,unicode,colorlinks=true}{hyperref}

\newif\ifanonymous

\anonymousfalse

\ifanonymous
\documentclass[acmsmall,screen,review,anonymous,nonacm]{acmart}
\else
\documentclass[acmsmall,screen]{acmart}
\fi

\newif\ifarXiv
\arXivtrue

\ifarXiv\else
\paperwidth=\dimexpr \paperwidth + 6cm\relax
\oddsidemargin=\dimexpr \oddsidemargin + 2.9cm\relax
\evensidemargin=\dimexpr \evensidemargin + 2.9cm\relax
\marginparwidth=\dimexpr \marginparwidth + 3cm\relax
\fi

\ifarXiv\PassOptionsToPackage{disable}{todonotes}\fi

\usepackage[ruled,vlined,linesnumbered]{algorithm2e}
\SetArgSty{textrm}
\SetKwComment{tcp}{$\rhd\;$}{}%
\makeatletter
\patchcmd{\@algocf@start}
{-1.5em}
{0pt}
{}{}
\makeatother

\makeatother
\usepackage{threeparttable}
\usepackage{todonotes,xspace}
\usepackage{xcolor}
\definecolor{shadecolor}{rgb}{0.90, 0.90, 0.90}

\newcommand{\blue}[1]{\textcolor{blue}{#1}}

\usepackage[many]{tcolorbox}

\definecolor{promptbg}{RGB}{245,245,245}    
\definecolor{titlebg}{RGB}{235,235,235}      

\newtcolorbox{promptbox}[1]{
    enhanced,
    breakable,
    colback=promptbg,        
    colbacktitle=titlebg,    
    colframe=titlebg,
    coltitle=black,          
    fonttitle=\bfseries\text\small,
    fontupper=\text\small,
    title=#1,               
    boxrule=1.5pt,          
    arc=1pt,                
    left=10pt,              
    right=10pt,             
    top=4pt,                
    bottom=2pt,             
    toptitle=1pt,           
    bottomtitle=1pt,        
    titlerule=0.5pt,        
    titlerule style={titlebg} 
}

\usepackage{graphicx}
\usepackage{caption}
\usepackage{subcaption}
\usepackage{amsmath}

\usepackage{circledsteps}

\renewcommand{\Circled}[1]{\textcircled{\raisebox{-.87pt}{\texttt{#1}}}}

\usepackage{mathtools}
\usepackage{csquotes}
\usepackage{wrapfig}

\usepackage{fancyvrb}
\usepackage[cachedir=minted-cache]{minted} 

\setminted[c]{
    frame=lines,
    framesep=2mm,
    baselinestretch=1.0,
    fontsize=\footnotesize, 
    breaklines=true,
    xleftmargin=0em,
    xrightmargin=0em,
    autogobble,
    escapeinside=||, 
    showspaces=false,
    showtabs=false
}

\usepackage{framed}

\usepackage{adjustbox}
\usepackage{makecell}
\usepackage{tikz}
\usetikzlibrary{patterns,calc,arrows,shapes,snakes,automata,backgrounds,petri,positioning,decorations.pathmorphing,intersections}
\usepackage{tkz-euclide}
\tikzset{mynode/.style={draw,solid,circle,inner sep=1pt}}
\usepackage[customcolors]{hf-tikz}
\hfsetfillcolor{gray!10}
\hfsetbordercolor{none}
\usepackage{pgfplots}
\usepgfplotslibrary{units}  
\pgfplotsset{compat=1.17}
\usepackage{multirow}
\usepackage{pifont} 

\usepackage{hyperref}
\usepackage[
    type={CC},
    modifier={by-nc-sa},
    version={3.0},
]{doclicense}

\usepackage[capitalise]{cleveref}
\crefname{equation}{}{}
\crefname{figure}{Fig.}{Figs.}

\crefdefaultlabelformat{#2\textup{#1}#3}

\newcommand{\wzycomment}[1]{\todo[color=blue!25,size=\small,fancyline,author=WZY]{#1}\xspace}

\newcommand{\ours}{\textup{\textsc{Accord}}\xspace}
\newcommand{\vst}{VST/Rocq\xspace}

\newcommand{\defeq}{{}\triangleq{}}

\def\triangleforqed{\hbox{$\lhd$}}
\makeatletter
\DeclareRobustCommand{\qedT}{%
	\ifmmode
	\eqno \def\@badmath{$$}
	\let\eqno\relax \let\leqno\relax \let\veqno\relax
	\hbox{\triangleforqed}%
	\else
	\leavevmode\unskip\penalty9999 \hbox{}\nobreak\hfill
	\quad\hbox{\triangleforqed}%
	\fi
}
\makeatother

\theoremstyle{plain}
\newtheorem*{thm*}{Theorem}
\newtheorem*{exmp*}{Example}
\newtheorem*{counterexmp*}{Counterexample}

\newtheoremstyle{ourstyle}{}{}{\itshape}{}{\bfseries}{.}{ }{\thmname{#1}\thmnumber{ #2}\thmnote{ (#3)}}
\theoremstyle{ourstyle}

\newtheorem{theorem}{Theorem}
\newtheorem{definition}[theorem]{Definition}

\newtheoremstyle{exmpstyle}{}{}{}{}{\bfseries}{.}{ }{\thmname{#1}\thmnumber{ #2}\thmnote{ (#3)}}
\theoremstyle{exmpstyle}

\newtheorem{example}[theorem]{Example}

\newtheoremstyle{rmkstyle}{}{}{}{}{\bfseries}{.}{ }{\thmname{#1}\thmnote{ (#3)}}
\theoremstyle{rmkstyle}

\newcounter{challenge}
\renewcommand{\thechallenge}{\arabic{challenge}} 

\crefname{challenge}{Challenge}{Challenges}   
\Crefname{challenge}{Challenge}{Challenges}   

\newenvironment{challenge}[1][]
{%
  \refstepcounter{challenge}
  \begin{tcolorbox}[boxrule=1pt,colback=white,colframe=black!75,boxsep=0mm,#1]%
  \textbf{Challenge \thechallenge.}\ %
}
{%
  \end{tcolorbox}%
}

\newcounter{insight}
\renewcommand{\theinsight}{\arabic{insight}} 

\crefname{insight}{Insight}{Insights}   
\Crefname{insight}{Insight}{Insights}   

{%
  \refstepcounter{insight}
  \begin{tcolorbox}[boxrule=1pt,colback=white,colframe=black!75,boxsep=0mm,#1]%
  \textbf{Insight \theinsight.}\ %
}
{%
  \end{tcolorbox}%
}

\providecommand*\phantomsection{}

\AtBeginDocument{%
  }

\makeatletter
\renewcommand\paragraph{\@startsection{paragraph}{4}{\z@}%
  {-.5\baselineskip \@plus -2\p@ \@minus -.2\p@}%
  {-3.5\p@}%
  {\ACM@NRadjust{\@parfont\@adddotafter}}}
\makeatother

\makeatletter
\ifx\correspondingauthor\undefined
\newif\if@ACM@corresponding@present
\@ACM@corresponding@presentfalse
\newcount\@ACM@corresponding@count
\newcommand\correspondingauthor{%
  \if@ACM@anonymous\else
    \@ACM@corresponding@presenttrue
    \advance\@ACM@corresponding@count by 1\relax
    \g@addto@macro\addresses{\@correspondingauthormark}%
  \fi}
\def\@correspondingauthormark{\g@addto@macro\@currentauthors{%
    \advance\hfuzz by 5pt\relax\textsuperscript{*}\relax}}
\def\@titlenotes{%
  \if@ACM@corresponding@present
    \footnotetext[1]{Corresponding author\ifnum\@ACM@corresponding@count>1 s\fi}%
  \fi}
\fi
\g@addto@macro\@titlenotes{%
  \if@ACM@anonymous\else
    \footnotetext[2]{Both authors contributed equally to this research.}%
  \fi}
\makeatother

\setcopyright{none}
\acmYear{2026}
\acmJournal{PACMPL}

\begin{document}

\title[Function Contracts and Beyond: Machine-Checked Assurance Cases]{Function Contracts and Beyond: Machine-Checked Assurance Cases for Conformance of C Libraries to Their Requirements}


\ifanonymous

\author{}

\else

\author{Haokun Li}
\authornotemark[2] 
\orcid{0000-0001-6411-9324}
\affiliation{%
  \institution{Zhejiang University}
  \city{Hangzhou}
  \country{China}
  \department{College of Computer Science and Technology}}
\email{ker@pm.me}

\author{Zhongyi Wang}
\orcid{0009-0008-1986-6070}
\authornotemark[2]
\affiliation{%
  \institution{Zhejiang University}
  \city{Hangzhou}
  \country{China}
  \department{College of Computer Science and Technology}}
\email{zhongyi.wang@zju.edu.cn}

\author{Mingshuai Chen}
\orcid{0000-0001-9663-7441}
\correspondingauthor
\affiliation{%
  \institution{Zhejiang University}
  \city{Hangzhou}
  \country{China}
  \department{College of Computer Science and Technology}}
\email{m.chen@zju.edu.cn}

\author{Ling Zhang}
\orcid{0000-0001-7190-6983}
\affiliation{
  \institution{Shanghai Jiao Tong University}
  \city{Shanghai}
  \country{China}
}
\email{ling.zhang@sjtu.edu.cn}

\author{Guanyan Li}
\orcid{0000-0002-4163-1840}
\affiliation{%
  \institution{University of Oxford}
  \city{Oxford}
  \country{UK}}
\email{guanyan.li.cs@gmail.com}

\author{Depeng Liu}
\orcid{0000-0002-9353-9691}
\affiliation{%
  \institution{Institute of Software, Chinese Academy of Sciences}
  \city{Beijing}
  \country{China}
  \department{Key Laboratory of System Software}}
\email{liudp@ios.ac.cn}

\author{Xiao Yi}
\orcid{0000-0002-4792-4433}
\affiliation{%
  \institution{The Chinese University of Hong Kong}
  \city{Hong Kong}
  \country{China}}
\email{yixiao5428@link.cuhk.edu.hk}

\author{Shengchao Qin}
\orcid{0000-0003-3028-8191}
\correspondingauthor
\affiliation{%
  \institution{Xidian University}
  \city{Xi'an}
  \country{China}}
\email{shengchao.qin@gmail.com}

\author{Jianwei Yin}
\orcid{0000-0003-4703-7348}
\affiliation{%
  \institution{Zhejiang University}
  \city{Hangzhou}
  \country{China}
  \department{College of Computer Science and Technology}}
\email{zjuyjw@zju.edu.cn}

\fi

\renewcommand{\shortauthors}{Mingshuai Chen, Zhongyi Wang and Haokun Li}

\begin{abstract}
Large language models (LLMs) have shown promise in automating interactive theorem proving, yet verification of real-world C codebases requires more than discharging individual proof goals.
The task involves jointly constructing expressive function specifications and their proofs, and ensuring that library interfaces compose along intended call sequences even without a designated client.
We present \ours, an LLM-assisted framework for building machine-checked \emph{assurance cases}: structured, auditable artifacts supporting the claim that a C codebase meets its intended requirements.
To model intended cross-interface use in open libraries, \ours constructs an \emph{interface protocol} that exposes permitted call sequences and resource assumptions for review, with a conditional safety guarantee under verified contracts and caller obligations.
\ours coordinates two complementary phases: (i) requirement-guided analysis and bottom-up construction of candidate specifications and protocols; and (ii) modular proof construction with feedback that revises the specifications and proofs.
Implemented using VST in Rocq, \ours establishes memory safety for 259 function definitions across seven implementation cases from six benchmark families, including industrial cryptographic modules, and additional functional correctness for red--black trees and ML-DSA.
Human review took an estimated 0.17--8.00 hours per case.
The guarantees depend on disclosed contracts and assumptions; human review supplies the conformance judgments connecting the formal artifacts to the intended requirements.
\end{abstract}




\maketitle

\setlength{\floatsep}{.8\baselineskip}
\setlength{\textfloatsep}{.8\baselineskip}
\setlength{\intextsep}{.8\baselineskip}
\captionsetup{aboveskip=5pt, belowskip=0pt}


%

\section{Introduction}\label{sec:introduction}

C codebases continue to carry the weight of safety-critical systems, and their formal verification rests on a machine-checked foundation: \emph{Hoare-style deductive verification} represents intended behavior by function \emph{contracts} (a.k.a., \emph{specifications}) and proves that implementations satisfy them~\cite{hoare-logic}.
Interactive theorem provers (ITPs) support a proof-producing style of verification: an ITP justifies a formal statement that the code meets its stated contracts, yielding machine-checked evidence for that statement~\cite{compcert-cacm09,sel4-sosp09}.
Such verification has provided assurance of critical systems, but developing the necessary contracts and proofs remains costly human work~\cite{sel4-tocs14,compcert-cacm09}.

Recent work has made substantial progress in using large language models (LLMs) to automate proof construction (and repair) in ITPs~\cite{baldur-fse23,proverbot9001,coqpilot-ase24,palm-ase24,cobblestone-icse26,rango-icse25,synver-ase25,slvc-icml26}.
Yet proof-construction automation alone does not meet the verification needs of a large real-world C codebase.
Typically it necessitates establishing an \emph{end-to-end assurance case}~\cite{iso-15026,assurance-case}: not only formal evidence that the implementation satisfies its specifications, but also an auditable argument that those specifications, taken together, conform to the intended requirements.
This raises two challenges at different levels.

First, \emph{how to construct function-level verification---the specification of a function and a machine-checked proof that its body satisfies the specification?}
The two obligations are \emph{coupled}: the concrete form of a specification---what it promises and how it is phrased---decisively shapes whether a body proof can go through, and proof attempts in turn expose where the specification is too strong, too weak, or mis-expressed.
A usable, provable pair is therefore the outcome of an iterative co-evolution, not of a single synthesis pass.
Proof-construction methods~\cite{baldur-fse23,proverbot9001,coqpilot-ase24,palm-ase24,cobblestone-icse26,rango-icse25,synver-ase25,slvc-icml26} take specifications as given, while contract-synthesis methods~\cite{preguss,cav24-autospec,icse25-specgen} typically validate their candidates with automated theorem provers rather than generate implementation proofs in an ITP.

Second, \emph{how to establish library-level conformance---that the exported interfaces, taken together, conform to the intended requirements?}
Here, a \emph{library} is a collection of modules, and its \emph{interfaces} refer to the specifications of its externally callable functions (e.g., typically the module-entry functions).
An \emph{open library} has no designated client functions (e.g., a \texttt{main}), a category that a substantial part of real-world codebases falls into~\cite{VSU-esop21,permissive-interface-fse05}.
Verifying an open library therefore raises a common difficulty:
proofs can establish that each module satisfies its interface, but \emph{library-level requirements also constrain how calls compose and how memory resources (e.g., allocated objects and borrowed buffers) and their contents evolve across them}. Per-module proofs do not per se show that the interfaces support the intended uses with the required guarantees.
Yet most LLM-assisted methods stop at function-level verification, leaving whether the library's interfaces, taken together, conform to the requirements an open problem that typically becomes a manual-review obligation. 
To address these challenges, we present \ours, an LLM-assisted framework that builds machine-checked assurance cases for the conformance of C libraries to their requirements.
For an open library, \ours constructs an \emph{interface protocol} that describes the permitted calls, their possible outcomes, and the resources carried between calls.
Its proof obligations require each interface to establish the facts and resources that subsequent calls need, guiding the synthesis of composable interfaces, and it exposes the permitted call sequences, with their guarantees and assumptions, for review.
Its soundness ensures conditional safety, i.e., protocol-following client executions incur no runtime errors in library calls, if the implementations satisfy contracts supporting the protocol.

\ours meets the first challenge by coordinating specification synthesis and proof construction in two phases.
The first phase derives natural-language function contracts top-down from the requirements and caller demands, so that each function promises what its callers need; it then formalizes these contracts bottom-up and constructs a candidate interface protocol.
The second phase constructs modular proofs that the function bodies satisfy their contracts and that these contracts support the protocol.
Proof attempts that expose gaps in the specifications or the protocol trigger their revision, so the initial candidates are not treated as fixed.
The resulting \emph{co-evolution} also connects function-level verification to the library-level requirements captured by the protocol.

We implement \ours using VST in Rocq and evaluate seven cases from six benchmark families, including three modules of the industrial cryptographic library openHiTLS~\cite{openhitls-source}.
The proofs establish memory safety for all 259 selected functions under disclosed contracts and assumptions, protocol-based interface composition for the libraries, and functional correctness for the red--black tree and openHiTLS's ML-DSA.
Machine checking validates the formal evidence, while human review (an estimated 0.17--8.00 hours per case) supplies the conformance judgments linking the formal artifacts to the requirements.
In openHiTLS's SLH-DSA, reviewing a call-order restriction in a synthesized contract also revealed an out-of-bounds read, since fixed upstream (\cref{sec:case-study}).

\paragraph{\bf Contributions.}
In summary, we make the following contributions:
\begin{itemize}
\item An LLM-assisted framework that builds end-to-end assurance cases by co-evolving specifications and proofs in two phases, with explicit human conformance review.
\item A resource-aware interface protocol for open libraries, whose proof obligations guide the synthesis of composable interfaces and whose soundness theorem guarantees the absence of runtime errors in protocol-following executions.
\item An evaluation on seven cases from six benchmark families, establishing memory safety for 259 functions, protocol-based interface composition, and case-specific functional correctness, while requiring at most one person-day (8 hours) of human review per case.
\end{itemize}

\paragraph{Paper Structure.}
After the motivating examples (\Cref{sec:motivation}), we present the assurance case (\Cref{sec:assurance-case}), the interface protocol (\Cref{sec:central-protocol}), and the two-phase {\ours} framework (\Cref{sec:framework}).
We then report the evaluation (\Cref{sec:evaluation}), discuss limitations (\Cref{sec:limitations}) and related work (\Cref{sec:related-work}), and conclude (\Cref{sec:conclusion}).
The appendices provide protocol examples and further formal and experimental details.


\begin{figure}[t]
\centering
\renewcommand{\arraystretch}{0.88}
\captionsetup{skip=3pt}
\begin{subfigure}{0.47\columnwidth}
\centering
\fbox{\ttfamily\footnotesize
\begin{tabular}{@{}l@{}}
  \texttt{\textbf{struct} node \{ \textbf{int} val; \textbf{struct} node *next; \}}\\
  \texttt{\textbf{struct} list \{ \textbf{struct} node *head; \}}\\
  \\
  \texttt{\textbf{int} list\_length(\textbf{struct} list *l) \{ }\\
  \texttt{\ \ \textbf{int} c = 0; \textbf{struct} node *cur = l->head;}\\
  \texttt{\ \ \textbf{while} (cur != NULL) \{ }\\
  {\ \ \ \ \setlength{\fboxsep}{0pt}\colorbox{orange!24}{\texttt{c = c + 1; \ /*}\ \normalfont overflow risk\ \texttt{*/}}}\\
  \texttt{\ \ \ \ cur = cur->next;}\\
  \texttt{\ \ \}}\\
  \texttt{\ \ \textbf{return} c;}\\
  \texttt{\}}\\
\end{tabular}}
\caption{The \texttt{list\_length} function in C.}
\label{fig:motivation-1-function}
\end{subfigure}
\hfill
\begin{subfigure}{0.52\columnwidth}
\centering
\fbox{\ttfamily\footnotesize
\begin{tabular}{@{}l@{\hspace{0em}}l@{}}
\multicolumn{2}{l}{\ttfamily (* \normalfont two versions: the {\smash{\setlength{\fboxsep}{0pt}\colorbox{red!18}{initial (-)}}} and the {\smash{\setlength{\fboxsep}{0pt}\colorbox{green!18}{revised (+)}}} one\ttfamily\ *)}\\
  & \textbf{Definition} list\_length\_spec := \\
  & \ \ DECLARE \_list\_length \\
  & \ \ WITH l: val, xs: list Z \\
  & \ \ PRE  [ tptr t\_list ] \\
  \multicolumn{2}{l}{{\setlength{\fboxsep}{0pt}\colorbox{red!18}{\hspace{0.2em}\texttt{-}\hspace{0.7em}\ PROP ()\ \ \ \ \ \ \ \ \ \ \ \ \ \ \ \ \ \ \ \ \ \ \ \ \ \ \ \ \ \ \ \ \ \ \ \ }}} \\
  \multicolumn{2}{l}{{\setlength{\fboxsep}{0pt}\colorbox{green!18}{\hspace{0.2em}\texttt{+}\hspace{0.7em}\ PROP (Zlength xs <= Int.max\_signed)\ \ \ \ \ \ \ }}} \\
  & \ \ \ \ \ PARAMS (l) GLOBALS () SEP (list\_rep l xs)\\
  & \ \ POST [ tint ] PROP () \\
  & \ \ \ \ \ RETURN (Vint (Int.repr (Zlength xs)))\\
  & \ \ \ \ \ SEP (list\_rep l xs). \\
\end{tabular}}
\caption{The \texttt{list\_length} specifications in {\vst}.}
\label{fig:motivation-1-spec}
\end{subfigure}
\\[0.4ex]
\begin{subfigure}{1.0\columnwidth}
\centering
\fbox{\ttfamily\footnotesize
\begin{tabular}{@{}l@{}}
  \textbf{Lemma} verify\_list\_length : semax\_body Vprog Gprog f\_list\_length list\_length\_spec.\\
  \textbf{Proof}.\ start\_function.\ forward.\ (* c = 0 *) forward.\ (* cur = l->head *) \\
  forward\_loop (EX front rest n,\ PROP (xs = front ++ rest; Int.signed n = Zlength front)\\
  \ \ \ \ SEP (seg h front cur; node\_chain cur rest)) break: (EX rest n,\ PROP (...)).\%assert\\
  \ \ \ {\ttfamily (* \normalfont loop goals: establishment / preservation / exit\ttfamily\ *)}\\
  \ \ \ - {\ttfamily (* \normalfont establishment\ttfamily\ *)} \{...\}\\
  \ \ \ - {\ttfamily (* \normalfont preservation\ttfamily\ *)} forward\_if.\\
  \ \ \ \ \ + {\ttfamily (* cur $\neq$ NULL\ttfamily\ *)} {\setlength{\fboxsep}{0pt}\colorbox{orange!24}{forward.\ {\ttfamily (* c = c + 1\ttfamily\ *)} \{...\}}}\ {\ttfamily (* \normalfont overflow check\ttfamily\ *)}\\
  \ \ \ \ \ \ \ forward.\ {\ttfamily (* cur = cur->next\ttfamily\ *)} \{...\}\\
  \ \ \ \ \ + {\ttfamily (* cur $=$ NULL\ttfamily\ *)} forward.\ {\ttfamily (* \normalfont loop exit\ttfamily\ *)} \{...\}\\
  \ \ \ - {\ttfamily (* \normalfont exit: reconstruct the list and return its length\ttfamily\ *)} \{...\}\\
Qed.
\end{tabular}}
\caption{The simplified proof in {\vst}; the \setlength{\fboxsep}{0pt}\colorbox{orange!24}{highlighted tactics} discharge the overflow risk at \texttt{c = c + 1}.}
\label{fig:motivation-1-proof}
\end{subfigure}
\caption{Function-level verification of \texttt{list\_length}: (a) the function body which contains a potential overflow, (b) its specifications in two versions (with the \setlength{\fboxsep}{0pt}\colorbox{red!18}{initial} and \setlength{\fboxsep}{0pt}\colorbox{green!18}{revised} differences highlighted), and (c) the proof against its revised specification.}
\label{fig:motivation-1}
\end{figure}

\section{Background and Motivation}\label{sec:motivation}

\subsection{Function-Level Verification}\label{sec:motivation-1-func-verif}

Function-level verification is the process of producing, jointly, a \emph{specification} that states what the given function is meant to do and a machine-checked \emph{proof} that the body satisfies that specification.
For example, the function (\cref{fig:motivation-1-function}) walks a linked list and returns its length.
\Cref{fig:motivation-1-spec} shows the shape of a specification in {\vst}: it binds a list of logical variables (\texttt{WITH}) and then states a precondition (\texttt{PRE}) and a postcondition (\texttt{POST}).
Within the pre- and postconditions, \texttt{PROP} collects pure logical facts, while \texttt{SEP} combines memory-resource predicates using separating conjunction.
The \texttt{PARAMS} clause specifies argument values, \texttt{GLOBALS} binds the global environment, and \texttt{RETURN} specifies the return value.
The proof in \cref{fig:motivation-1-proof} establishes a Hoare-style judgment. 
Starting from the precondition, VST's \emph{forward symbolic execution} applies program-logic rules to derive assertions describing successive program points.
Tactics such as \texttt{forward}, \texttt{forward\_if}, and \texttt{forward\_loop} organize the proof, generating obligations for statements, branches, and loops, and ultimately establishing the postcondition on return.
Rocq checks the resulting proof.

Function-level verification in VST also establishes \emph{execution safety}, conditional on the function precondition and the contracts assumed for its callees.
This includes \emph{memory safety}---the absence of invalid memory accesses and deallocations---and the exclusion of other runtime errors covered by the Clight/VST verification model~\cite{compcert-cacm09,VST-esop11}.
Appendix~\ref{app:vst-safety} summarizes this coverage of memory and arithmetic errors and lists where the model differs from ISO C~\cite{c-standard}.

The bottleneck of function-level verification lies in \emph{the inherent coupling of constructing the specification and the proof}. The nature of this coupling is best seen through a concrete failure.
Take the \setlength{\fboxsep}{0pt}\colorbox{red!18}{initial} contract of \cref{fig:motivation-1-spec}, which imposes no pure-logic constraints on the input list (i.e., \setlength{\fboxsep}{0pt}\colorbox{red!18}{\texttt{PROP ()}} in the precondition).
It looks unobjectionable: the implementation obviously counts the elements, and the postcondition states exactly this. The proof, however, does not go through.
\texttt{forward\_loop} on the \texttt{while} loop yields the three standard obligations---establishment, preservation, and exit---which are discharged routinely.
The trouble appears at the assignment \setlength{\fboxsep}{0pt}\colorbox{orange!24}{\texttt{c = c + 1}}: here \setlength{\fboxsep}{0pt}\colorbox{orange!24}{\texttt{forward}} generates a signed-overflow proof goal, requiring that the incremented value still fit in an \texttt{int}, since an overflowing addition is undefined behavior (UB) in C. To proceed, one must \setlength{\fboxsep}{0pt}\colorbox{green!18}{revise} the specification, strengthening its precondition with the bound (\setlength{\fboxsep}{0pt}\colorbox{green!18}{\texttt{Zlength xs <= Int.max\_signed}}); under this refinement the overflow goal closes and an ordinary loop proof succeeds. \emph{The specification is thus shaped by the proof, and the proof can advance only as far as the specification permits---the two co-evolve, each constraining and refining the other.}

The difficulty is far more severe in real-world verification, where contracts are far richer than a single bound and the point at which a proof stalls is seldom as local as one assignment: determining whether an obstruction calls for a stronger hypothesis, a sharper invariant, a different specification shape, or a genuinely weaker requirement is a judgment made against the proof in flight.
The challenge is thus not to produce a specification or a proof in isolation, with the other taken as given (which is what existing works~\cite{preguss,cav24-autospec,icse25-specgen,baldur-fse23,proverbot9001,coqpilot-ase24,palm-ase24,cobblestone-icse26,rango-icse25,synver-ase25,slvc-icml26} focus on), but to co-evolve them:
\begin{challenge}
\label{cha:function-level-verification}
How to co-evolve the specification and the proof, so as to verify each function?
\end{challenge}

\begin{figure}[t]
\centering
\renewcommand{\arraystretch}{0.88}
\captionsetup{skip=3pt}
\begin{subfigure}{1.0\columnwidth}
\centering
\setlength{\fboxsep}{0.2em}\fbox{\ttfamily\footnotesize
\begin{tabular}{@{}l@{}}
  \texttt{struct list *list\_new(void);} \hfill \texttt{/* allocate an empty list */} \\
  \texttt{void list\_free(struct list *l);} \hfill \texttt{/* release a list */} \\
  \texttt{void list\_push(struct list *l, int v);} \hfill \texttt{/* insert a node with value v at the head; length +1 */} \\
  \texttt{void list\_remove(struct list *l, int v);} \hfill \texttt{/* delete a node with v; do nothing if missed */} \\
  \texttt{int list\_length(struct list *l);} \hfill \texttt{/* return the length of l, which is required to <= INT\_MAX */} \\
\end{tabular}}
\caption{The five function declarations of the linked-list library (definitions elided).}
\label{fig:motivation-2-library}
\end{subfigure}
\\[0.4ex]
\begin{subfigure}{0.465\columnwidth}
\centering
\setlength{\fboxsep}{0.2em}\fbox{\ttfamily\footnotesize
\begin{tabular}{@{}l@{}}
  \textbf{Definition} list\_push\_spec := \\
  \ \ DECLARE \_list\_push \\
  \ \ WITH l: val, v: Z, xs: list Z, gv: globals \\
  \ \ PRE  [ tptr t\_list; tint ] \\
  \ \ \ PROP (repable\_signed v) \\
  \ \ \ PARAMS (l; Vint (Int.repr v)) GLOBALS (gv) \\
  \ \ \ SEP (list\_rep l xs; mem\_mgr M gv) \\
  \ \ POST [ tvoid ] PROP (...) \\
  \multicolumn{1}{@{}l}{{\setlength{\fboxsep}{0pt}\colorbox{red!18}{\hspace{0.2em}\texttt{-}\hspace{0.4em}\ SEP (list\_rep l ys; mem\_mgr M gv)\ \ \ \ \ \ \ }}} \\
  \multicolumn{1}{@{}l}{{\setlength{\fboxsep}{0pt}\colorbox{green!18}{\hspace{0.2em}\texttt{+}\hspace{0.4em} SEP (list\_rep l (v :: xs); mem\_mgr M gv)}}} \\
\end{tabular}}
\caption{The simplified \texttt{list\_push} specifications.}
\label{fig:motivation-2-push-spec}
\end{subfigure}
\hfill
\begin{subfigure}{0.525\columnwidth}
\centering
\setlength{\fboxsep}{0.2em}\fbox{\ttfamily\footnotesize
\begin{tabular}{@{}l@{}}
  \textbf{Definition} list\_remove\_spec := \\
  \ \ DECLARE \_list\_remove \\
  \ \ WITH l: val, v: Z, xs: list Z, gv: globals \\
  \ \ PRE  [ tptr t\_list; tint ] \\
  \ \ \ PROP (repable\_signed v) \\
  \ \ \ PARAMS (l; Vint (Int.repr v)) GLOBALS (gv) \\
  \ \ \ SEP (list\_rep l xs; mem\_mgr M gv) \\
  \ \ POST [ tvoid ] PROP (...) \\
  \multicolumn{1}{@{}l}{{\setlength{\fboxsep}{0pt}\colorbox{red!18}{\hspace{0.2em}\texttt{-}\hspace{0.4em}\ SEP (list\_rep l ys; mem\_mgr M gv)\ \ \ \ \ \ \ \ \ \ \ \ \ }}} \\
  \multicolumn{1}{@{}l}{{\setlength{\fboxsep}{0pt}\colorbox{green!18}{\hspace{0.2em}\texttt{+}\hspace{0.4em} SEP (list\_rep l (le\_or\_id v xs); mem\_mgr M gv)}}} \\
\end{tabular}}
\caption{The simplified \texttt{list\_remove} specifications.}
\label{fig:motivation-2-remove-spec}
\end{subfigure}
\caption{The linked-list library exposes five functions (a). Both versions of the \texttt{list\_push} and \texttt{list\_remove} specifications (b)(c), with the \setlength{\fboxsep}{0pt}\colorbox{red!18}{initial} and \setlength{\fboxsep}{0pt}\colorbox{green!18}{revised} differences highlighted, are provable; only the revised (stronger) ones conform to the intended requirements.}
\label{fig:motivation-2}
\end{figure}

\subsection{Library-Level Conformance}\label{sec:motivation-2-inter-conform}

Suppose now that the difficulty of \Cref{sec:motivation-1-func-verif} has been overcome: every function of the library is proved to satisfy its contract. Yet such function-level verification does not guarantee that the interfaces, taken together, conform to the intended requirements, i.e., \emph{library-level conformance}. These requirements may span multiple interface calls\footnote{``Interface calls'' denotes calls to the corresponding externally callable functions; an interface itself is a function specification.}: the resources and facts established by one call must enable later calls to execute safely and produce the intended results. In a {\vst} proof, each call's precondition must be established from the preceding calls' postconditions.

The library in \cref{fig:motivation-2-library} has five interfaces. \texttt{list\_length} is verified against the revised contract of \cref{fig:motivation-1-spec}, \texttt{list\_push} and \texttt{list\_remove} are proved to satisfy their specifications in \cref{fig:motivation-2-push-spec,fig:motivation-2-remove-spec}, and the specifications of \texttt{list\_new} and \texttt{list\_free} are elided. Consider the client function in \Cref{fig:motivation-client},
which follows a call sequence the design intends to be legal.
Yet in the specifications of \cref{fig:motivation-2-push-spec,fig:motivation-2-remove-spec}, where each function has been individually verified, this client function may still not be provable, due to a \setlength{\fboxsep}{0pt}\colorbox{orange!24}{precondition violation}. This is because the postconditions of \texttt{list\_push} and \texttt{list\_remove} are too weak to meet the precondition of \texttt{list\_length} at the assignment \setlength{\fboxsep}{0pt}\colorbox{orange!24}{\texttt{int n = list\_length(l)}}. Specifically, the separating spatial predicates (i.e.,  
\begin{wrapfigure}{r}{0.31\columnwidth}
\centering
\ttfamily\footnotesize
\begin{tabular}{@{}l@{}}
\texttt{void client() \{}\\
  \texttt{\ \ list *l = list\_new();}\\
  \texttt{\ \ list\_push(l, 1);}\\
  \texttt{\ \ list\_push(l, 2);}\\
  \texttt{\ \ list\_remove(l, 1);}\\
  \texttt{\ \ }\setlength{\fboxsep}{0pt}\colorbox{orange!24}{\texttt{/* }\normalfont precondition-violation risk\texttt{ */}}\\
  \texttt{\ \ \setlength{\fboxsep}{0pt}\colorbox{orange!24}{int n = list\_length(l);}}\\
  \texttt{\ \ list\_free(l);}\\
  \texttt{\}}\\
\end{tabular}
\caption{A legal client function.}
\label{fig:motivation-client}
\end{wrapfigure}
\setlength{\fboxsep}{0pt}\colorbox{red!18}{\texttt{SEP (list\_rep l ys; ...)}}) of the initial postconditions of \texttt{list\_push} and \texttt{list\_remove} state that the list at the returning point is \emph{some} list \texttt{ys}, unrelated to the input \texttt{xs}, while \texttt{list\_length} requires the length of the input length to fit in the int range (\cref{fig:motivation-1-spec}).
Both the initial interfaces in \cref{fig:motivation-2-push-spec,fig:motivation-2-remove-spec} are provable---take \texttt{ys} to be the actual result---yet a caller learns nothing about the current length of the list.
To discharge the \emph{precondition} of \texttt{list\_length}, one must accumulate the length along the sequence; from a postcondition that names an arbitrary \texttt{ys} this is impossible.
The \setlength{\fboxsep}{0pt}\colorbox{green!18}{revised} separating predicates in the postconditions repair this: \texttt{list\_push} returns \setlength{\fboxsep}{0pt}\colorbox{green!18}{\texttt{v :: xs}}---the list with \texttt{v} at the head, and \texttt{list\_remove} returns \setlength{\fboxsep}{0pt}\colorbox{green!18}{\texttt{le\_or\_id v xs}}---the list with its first \texttt{v} removed, or the list itself if \texttt{v} is absent.
Each is still provable of its body; the difference is that the revised strong forms carry exactly the length information a sequence of calls needs.

In an open library, requirements such as memory safety and functional correctness concern how interfaces may be called and how those calls affect resources and data.
Establishing library-level conformance therefore requires more than verifying each function: the interfaces must compose to justify the intended uses and their expected behavior.
Existing LLM-assisted methods for contract and proof generation~\cite{preguss,cav24-autospec,icse25-specgen,baldur-fse23,proverbot9001,coqpilot-ase24,palm-ase24,cobblestone-icse26,rango-icse25,synver-ase25,slvc-icml26} therefore face a further challenge:
\begin{challenge}
\label{cha:library-level-conformance}
How to construct interfaces that jointly conform to a library's requirements?
\end{challenge}


\section{The Assurance Case}\label{sec:assurance-case}

\Cref{sec:motivation} developed the two primary challenges in constructing an end-to-end \emph{assurance case}~\cite{iso-15026,assurance-case} for a C codebase.
This is precisely what {\ours} builds: \emph{a structured, auditable set of artifacts supporting the claim that the C codebase meets its design requirements}.

\Cref{fig:assurance-case} gives an overview of this assurance case. It has five components connected by four kinds of relations; the three components and two relations in the shaded region are \emph{formal artifacts} that {\ours} produces automatically and VST/Rocq checks.
The C codebase is the input under verification, while the requirements reflect the design intent.
Starting from the codebase, CompCert's front end mechanically translates the C source into a Clight AST~\cite{compcert-cacm09} (\Circled{1}).
{\ours} generates specifications for the Clight functions and proofs that the functions satisfy those specifications (\Circled{2}), which also establish execution safety under the stated preconditions and callee contracts (\Cref{sec:motivation-1-func-verif}).
For an open library, {\ours} also constructs an interface protocol to model intended interface-call sequences, and the resource and value relationships across them.
It proves that the function specifications subsume the corresponding protocol edge contracts (\Circled{3}).
Together, implementation satisfaction (\Circled{2}) and per-edge subsumption (\Circled{3}) support the protocol soundness theorem: executions following the protocol are safe under its stated assumptions (\Cref{thm:protocol-soundness}).
The remaining manual obligation in the assurance case is to review whether the function specifications and the interface protocol conform to the requirements (\Circled{4}), for instance whether a precondition is too strong or the protocol omits a legal and safe interface-call sequence.
We formalize the interface protocol in \cref{sec:central-protocol}, and explain how {\ours} constructs this assurance case in \cref{sec:framework}.

\begin{figure}[t]
\centering
\includegraphics[width=0.9\columnwidth]{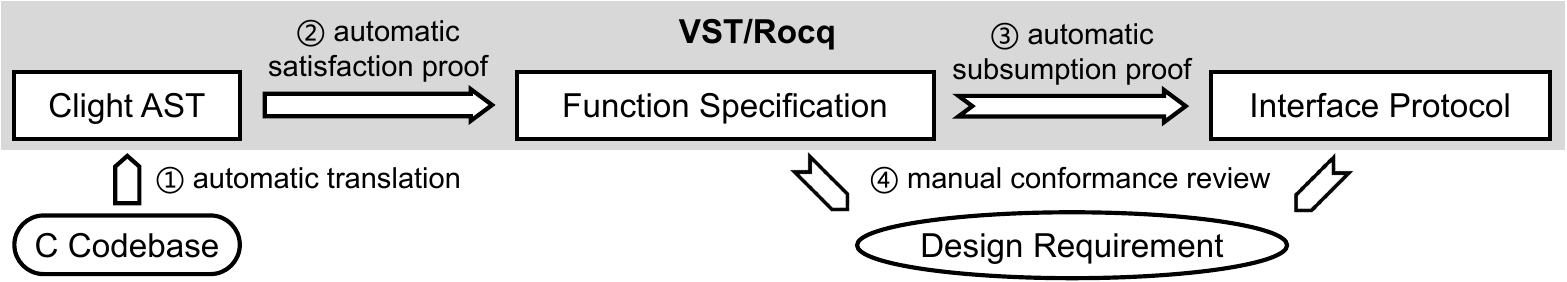}
\caption{The assurance case built by {\ours}.}
\label{fig:assurance-case}
\end{figure}

\Cref{fig:assurance-case} depicts the assurance case for an open library, which has no designated client within the verification scope.
When a library instead contains a client (e.g., a \texttt{main} function), that in-scope client supplies the concrete calling context that an open library lacks.
{\ours} therefore replaces the interface protocol with a client-level specification and whole-program evidence; the other four components and three relations (\Circled{1}\Circled{2}\Circled{4}), remain.

\paragraph{Trusted computing base (TCB)}
Our assurance relies on the correctness of the Rocq kernel, the Clight semantic model for the claimed C behavior, and the configured preprocessing and C-to-Clight translation.
The soundness of VST's program logic w.r.t\ the Clight semantics is machine-checked in Rocq.
Logical axioms and external-function premises are separately disclosed and reviewed.


\section{The Interface Protocol for Open Libraries}\label{sec:central-protocol}

As \Cref{sec:motivation-2-inter-conform} showed, function-level verification alone does not establish that a library's interfaces conform to its requirements.
Library-level requirements constrain how interfaces may be called, and which relationships among resources and values must hold across calls.
An open library has no designated client whose control flow could specify the intended uses.
Establishing \emph{library-level conformance} therefore involves two tasks: (i) identifying the interface-call sequences permitted by the requirements and (ii) proving that the corresponding interfaces compose along those sequences.
To support both tasks, \ours constructs an \emph{interface protocol}: a state-based model of how calls establish, preserve, and change the resources and facts needed by later calls. It exposes these relationships and assumptions for review and guides interface synthesis and refinement.

\subsection{Interface Protocol: States and Behaviors}\label{sec:formal-def}

\begin{definition}[Interface]\label{def:interface}
Let $\mathbb F$ be the set of externally callable functions in an open library. For each $f\in\mathbb F$, its \emph{interface} $I_f$ is the specification
\begin{equation}\label{eq:interface-schema}
  I_f \defeq
  \forall w\in W_f.\;
  \bigl\{P_f(w)\bigr\}\;
  f\bigl(\mathsf{args}_f(w)\bigr)\;
  \bigl\{\mathsf{ret}.\;Q_f(w,\mathsf{ret})\bigr\}, \quad \text{where}
\end{equation}
\[
  \mathsf{args}_f : W_f\to\mathsf{CArgs}_f,\quad
  P_f : W_f\to\mathsf{Assn},\quad
  Q_f : W_f\times\mathsf{Ret}_f\to\mathsf{Assn}.
\]
Here $\mathsf{CArgs}_f$ is the domain of concrete C argument tuples and $\mathsf{Ret}_f$ is the return-value domain of $f$, taken to be a singleton for void functions; the postcondition binds the actual return value as $\mathsf{ret}$.
We assume a sound separation logic for the target language.
Abstracting from the VST syntax of \Cref{sec:motivation-1-func-verif}, we write $\mathsf{Assn}$ for the assertions of this logic, $\mathsf{Prop}$ for pure propositions, and $\mathsf{MPred}$ for spatial assertions.
$W_f$ contains the logical inputs instantiated by the caller (c.f. \texttt{WITH} variables \cref{fig:motivation-1-spec}); $\mathsf{args}_f$ maps them to concrete C arguments.
We write the assertion $\ulcorner p\urcorner\land s$, with $p\in\mathsf{Prop}$ and $s\in\mathsf{MPred}$, as the pair $(p,s)$, and $X_{\mathsf p}$, $X_{\mathsf s}$ for the two parts of a pair $X$.
Quantifiers, $*$, and $\vdash$ are those of the logic; in particular $F*(p,s)=(p,F*s)$.
\end{definition}

\begin{example}\label{ex:length}
To instantiate \cref{eq:interface-schema} for \texttt{list\_length} in \cref{fig:motivation-1-spec}, take
$W_{\mathsf{length}}=\mathsf{Val}\times\mathsf{list}\ \mathbb Z$, $w=(l,xs)$, $\mathsf{args}_{\mathsf{length}}(l,xs)=(l)$, and
\[
\begin{gathered}
P_{\mathsf{length}}(w)=
  (\mathsf{Zlength}\ xs\leq\mathsf{Int.max\_signed},\;\mathsf{list\_rep}\ l\ xs),\\
Q_{\mathsf{length}}(w,\mathsf{ret})=
  (\mathsf{ret}=\mathsf{Zlength}\ xs,\;\mathsf{list\_rep}\ l\ xs).
\end{gathered}
\]
Here the bound $\mathsf{ret}$ is the C return value decoded as a mathematical integer.
The two components of $P_{\mathsf{length}}$ correspond to VST's precondition \texttt{PROP} and \texttt{SEP} clauses. For $Q_{\mathsf{length}}$, the pure component additionally incorporates the decoded \texttt{RETURN} constraint.
\qedT
\end{example}
\begin{definition}[Interface Protocol]\label{def:interface-protocol}
Given the interface functions $\mathbb F$ of an open library, its \emph{interface protocol} is the six-tuple
\begin{equation}\label{eq:protocol-tuple}
  \mathcal P
  \defeq
  (\mathbb F,\mathbb S,\llbracket\cdot\rrbracket,\mathbb G,\mathbb B,\mathsf{BorrowIn}).
\end{equation}
Let $\mathbb S_G$, $\mathbb S_B$, and $\mathbb S_L$ be the state domains for the library's global variables, resources currently borrowed from clients, and library-managed objects, respectively. The state space is
\begin{equation}\label{eq:protocol-state-space}
  \mathbb S
  \defeq
  \mathbb S_G\times\mathbb S_B\times\mathbb S_L.
\end{equation}
A state $\sigma=(\gamma,\beta,\ell)\in\mathbb S$ records the resources tracked by the protocol across interface calls, organized into the three components above. Its interpretation $\llbracket\cdot\rrbracket:\mathbb S\to\mathsf{Prop}\times\mathsf{MPred}$ is
\begin{equation}\label{eq:state-interpretation}
  \llbracket\sigma\rrbracket
    \defeq
      \bigl(\llbracket\sigma\rrbracket_{\mathsf p},
            \llbracket\sigma\rrbracket_{\mathsf s}\bigr),
  \ \text{where}\
  \llbracket(\gamma,\beta,\ell)\rrbracket_{\mathsf s}
    \defeq
      \mathsf{GlobalRep}(\gamma)
      * \mathsf{BorrowRep}(\beta)
      * \mathsf{LibObjRep}(\ell).
\end{equation}
The pure component $\llbracket\sigma\rrbracket_{\mathsf p}$ constrains the logical values recorded in $\sigma$ and their relationships. The spatial component $\llbracket\sigma\rrbracket_{\mathsf s}$ describes the memory resources that represent those values.\footnote{For example, a SHA-256 state may record the number $n$ of bytes currently buffered in a context. A pure constraint requires $0\leq n<64$, while the spatial component describes the context memory containing that value. Appendix~\ref{app:sha256-state-interpretation} develops this example. In the linked-list example below, no additional pure state constraints are imposed, so $\llbracket\sigma\rrbracket_{\mathsf p}=\mathsf{True}$.}
The three representation predicates and the pure component are part of the protocol.

A \emph{branch} describes one possible result of a call to $f$: the next protocol state, the specified return value, a pure fact about the result, and a residual resource predicate not represented by that state. Formally, a branch and its type are written as
\begin{equation}\label{eq:branch}
  br\defeq(br.\mathsf{next},br.\mathsf{ret},br.\mathsf{prop},br.\mathsf{res})
  \in\mathsf{Branch}_f,\quad
  \mathsf{Branch}_f\defeq
  \mathbb S\times\mathsf{Ret}_f\times\mathsf{Prop}\times\mathsf{MPred}.
\end{equation}
The resource predicate also describes the contents of out-parameter memory outside the successor state.
For each $f\in\mathbb F$, let $O_f$ be the domain of post-call logical witnesses, such as updated objects, newly allocated addresses, or returned values.
The last three components in \cref{eq:protocol-tuple} are the per-function families
\[
\begin{aligned}
  \mathbb G&=\{G_f\mid f\in\mathbb F\},
  &G_f&:\mathbb S\times W_f\to\mathsf{Prop},\\
  \mathbb B&=\{\mathsf{Br}_f\mid f\in\mathbb F\},
  &\mathsf{Br}_f&:\mathbb S\times W_f\times O_f\to 2^{\mathsf{Branch}_f},\\
  \mathsf{BorrowIn}&=\{\mathsf{BorrowIn}_f\mid f\in\mathbb F\},
  &\mathsf{BorrowIn}_f&:W_f\to\mathsf{MPred}.
\end{aligned}
\]
Here $G_f(\sigma,w)$ is the pure guard for calling $f$ from $\sigma$ with logical input $w$, $\mathsf{Br}_f(\sigma,w,o)$ gives its possible branches, and $\mathsf{BorrowIn}_f(w)$ is what the client lends for this call beyond the resources of the state.
Together with the state interpretation and $\mathsf{BorrowIn}_f(w)$, the guard must suffice to establish the interface precondition $P_f(w)$, as formalized by per-edge subsumption (\cref{def:edge-subsumption}).
\end{definition}


Designing interface protocols for real-world libraries is nontrivial, often involving multiple state variants and complex branching behavior.
Appendix~\ref{app:mceliece-protocol} illustrates this design process with Classic McEliece, a post-quantum cryptographic library (also evaluated in \cref{sec:evaluation}).
To make the basic concepts easier to follow, we use a highly simplified linked-list protocol in the following example.

\begin{example}\label{ex:linked-list-components}
For the simplified linked-list protocol of \cref{fig:motivation-2}, take $\mathbb S_G=\{\mu_{\mathsf{mgr}}\}$ for the allocator state, $\mathbb S_B=\{\emptyset\}$ for no tracked borrowed resources, and $\mathbb S_L$ to be maps from live handles to abstract lists, i.e.\ $\ell=\{p\mapsto xs,\ldots\}\in\mathbb S_L$.
Here $\mu_{\mathsf{mgr}}$ is simply a state label, interpreted by $\mathsf{mem\_mgr}\;M\;gv$, the logical resources required by the allocation and deallocation contracts, for a fixed allocator model $M$ and global environment $gv$.
No additional pure state constraint is imposed in this model, so $\llbracket\sigma\rrbracket_{\mathsf p}=\mathsf{True}$.
For the length query, $\mathsf{BorrowIn}_{\mathsf{length}}(w)=\mathsf{emp}$.
The spatial interpretation is
\[
\begin{gathered}
  \llbracket (\mu_{\mathsf{mgr}}, \emptyset, \ell)\rrbracket_{\mathsf s}
  \defeq \mathsf{mem\_mgr}\;M\;gv * \mathsf{emp} *
    \mathop{\mbox{\large $*$}}\limits_{p\mapsto xs\in \ell}
    \mathsf{list\_rep}\;p\;xs.
\end{gathered}
\]
Thus the state carries the allocator resource together with the concrete representations of all live library-managed lists.
For the \texttt{list\_length} interface of \Cref{ex:length}, take $O_{\mathsf{length}}=\mathbb Z$ and $\mathsf{Ret}_{\mathsf{length}}=\mathbb Z$ for the logical outcome and decoded return values, respectively, and define
\[
\begin{aligned}
  G_{\mathsf{length}}((\mu_{\mathsf{mgr}}, \emptyset, \ell),(p,xs))
    &\defeq
      \underbrace{p\mapsto xs\in \ell}_{\text{state membership}}
      \land
      \underbrace{\mathsf{Zlength}\;xs\leq\mathsf{Int.max\_signed}}_{\text{length bound}},\\
  \mathsf{Br}_{\mathsf{length}}(\sigma,(p,xs),o)
    &\defeq
      \begin{cases}
        \bigl\{(\sigma,o,\mathsf{True},\mathsf{emp})\bigr\} & \text{if } o=\mathsf{Zlength}\;xs,\\
        \emptyset & \text{otherwise.}
      \end{cases}
\end{aligned}
\]
State membership connects the chosen input $(p,xs)$ to the tracked list representation; the length bound supplies the pure condition of $P_{\mathsf{length}}$.
The sole feasible branch leaves the state unchanged, declares the return value as $\mathsf{Zlength}\;xs$, and adds no pure constraint or residual resource.
\qedT
\end{example}

We now use these components to define protocol transitions, specifying which calls are permitted at a state and which states may follow.
Composing transitions yields interface-call sequences, allowing the protocol to express requirements on the order and conditions of library calls.

\begin{definition}[Protocol Transition]\label{def:protocol-transition}
For $f\in\mathbb F$, $w\in W_f$, and $\sigma,\sigma'\in\mathbb S$, the protocol-transition judgment is defined by
\begin{equation}\label{eq:v3-transition}
  \sigma\xrightarrow{f(w)}\sigma'
  \quad\Longleftrightarrow\quad
  G_f(\sigma,w)\land
  \exists o\in O_f,\,br\in\mathsf{Br}_f(\sigma,w,o).\;
  br.\mathsf{next}=\sigma'.
\end{equation}
It states that a call to $f$ with logical input $w$ may take the protocol from $\sigma$ to $\sigma'$ if and only if the guard holds and a branch for some outcome has $\sigma'$ as its successor.
We call the pair $(\sigma,f)$ \emph{enabled} if
\[
  \exists w\in W_f,\;\sigma'\in\mathbb S.\quad
  \sigma\xrightarrow{f(w)}\sigma'.
\]
\end{definition}

\begin{example}\label{ex:linked-list-transition}
The \texttt{list\_length} query in~\Cref{ex:linked-list-components} gives the self-transition
\[
  \sigma\xrightarrow{\mathsf{length}(p,xs)}\sigma
  \quad\Longleftrightarrow\quad
  p\mapsto xs\in\ell
  \land
  \mathsf{Zlength}\;xs\leq\mathsf{Int.max\_signed}.\qedT
\]
\end{example}

\subsection{Protocol Edges and Subsumption}\label{sec:cp-edges}

The protocol-transition judgment declares an abstract state change via a function $f$, but does not show that the implementation of $f$ supports it.
Conversely, the interface $I_f$ does not explicitly describe how one call to $f$ transforms the resources represented by a particular state.
A \emph{protocol edge} connects the two levels by giving a specification of $f$ for calls made from a given state $\sigma$.

\begin{definition}[Protocol Edge]\label{def:protocol-edge}
For a state $\sigma$ and interface function $f$, define the \emph{edge}
\begin{equation}\label{eq:edge-schema}
  E_{\sigma,f}
  \defeq
  \forall w\in W_f.\;
  \bigl\{P^E_{\sigma,f}(w)\bigr\}\;
  f\bigl(\mathsf{args}_f(w)\bigr)\;
  \bigl\{\mathsf{ret}.\;Q^E_{\sigma,f}(w,\mathsf{ret})\bigr\},
\end{equation}
whose pre- and postconditions, unlike those of \cref{def:interface}, are bound to $(\sigma,f)$ rather than to $f$ alone.
The edge conditions are defined as follows:
\begin{equation}\label{eq:edge-conditions}
\begin{aligned}
  P^E_{\sigma,f}(w)
    &\defeq
      \bigl(\llbracket\sigma\rrbracket_{\mathsf p}\land G_f(\sigma,w),\;
            \llbracket\sigma\rrbracket_{\mathsf s} * \mathsf{BorrowIn}_f(w)\bigr),\\
  Q^E_{\sigma,f}(w,\mathsf{ret})
    &\defeq
      \exists o\in O_f,\,br\in\mathsf{Br}_f(\sigma,w,o).\\
    &\quad
      \bigl(\llbracket br.\mathsf{next}\rrbracket_{\mathsf p}
              \land br.\mathsf{prop}\land\mathsf{ret}=br.\mathsf{ret},\;
            \llbracket br.\mathsf{next}\rrbracket_{\mathsf s} * br.\mathsf{res}\bigr).
\end{aligned}
\end{equation}
The precondition combines $\llbracket\sigma\rrbracket$, the call guard, and $\mathsf{BorrowIn}_f(w)$.
If $w$ has no branch, the postcondition is false and the call must not return.
The postcondition describes a declared branch, i.e. $br$.
Its pure part re-establishes the successor state's pure constraints, includes the branch fact $br.\mathsf{prop}$, and constrains the actual return value to $br.\mathsf{ret}$.
Its spatial part provides the resources tracked by the successor, together with $br.\mathsf{res}$, including any out-parameter contents represented outside that state.
Resources supplied through $\mathsf{BorrowIn}_f$ may be retained in the successor state's $\mathbb S_B$ component or returned to the client through $br.\mathsf{res}$.
\end{definition}

A transition in~\cref{eq:v3-transition} records one possible successor, supported by at least one branch.
The edge $E_{\sigma,f}$ packages the guard and resources together with all declared branches for calls to $f$ from $\sigma$.
To justify this state-indexed abstraction of interface behavior under the edge's precondition, we require $I_f$ to subsume $E_{\sigma,f}$, as defined next.

\begin{definition}[Per-Edge Subsumption]\label{def:edge-subsumption}
We say that $I_f$ \emph{subsumes} $E_{\sigma,f}$, written $\mathsf{Sub}(I_f,E_{\sigma,f})$, if for each $w\in W_{f}$ there is a frame $F$ with
\begin{equation}\label{eq:subsumption-obligations}
  P^E_{\sigma,f}(w)\ \vdash\ F * P_f(w),
  \qquad
  \forall\mathsf{ret}\in\mathsf{Ret}_f.\;
  F * Q_f(w,\mathsf{ret})
  \ \vdash\ Q^E_{\sigma,f}(w,\mathsf{ret}).
\end{equation}
The first entailment establishes that an admissible edge call supplies the interface precondition; the second requires every interface outcome, recombined with the frame, to fit a declared branch of the edge.
In VST, \ours establishes the relation using \texttt{funspec\_sub}, which additionally handles logical-input matching and the required typing conditions.
\end{definition}

\subsection{Runs and Protocol Soundness}\label{sec:cp-soundness}

Per-edge subsumption relates one interface call to its protocol description.
We now extend this relation to sequences of interface calls.

\begin{definition}[Valid Run]\label{def:protocol-run}
A run of $\mathcal P$ is a finite or infinite sequence of states and calls, written
\[
  r=\sigma_0\xrightarrow{f_1(w_1)}\sigma_1
       \xrightarrow{f_2(w_2)}\sigma_2\cdots.
\]
It may start at any $\sigma_0\in\mathbb S$.
The run is \emph{valid}, written $\mathsf{ValidRun}_{\mathcal P}(r)$, if every displayed transition satisfies \cref{eq:v3-transition}; we also say that $\mathcal P$ \emph{accepts} $r$.
\end{definition}

\begin{definition}[Run Realization]\label{def:run-realization}
A run $r$ models a class of concrete client executions.
We consider sequential executions: library calls neither overlap nor nest.
A state \emph{satisfies} an assertion if some part of it does in the model of the underlying logic; \emph{preservation} means that the assertion continues to hold between successive library calls.
An execution $\pi$ of a client of the library $L$ \emph{realizes} $r$, written $\mathsf{Realizes}_{\mathcal P}(\pi,r)$, when
(i) Immediately before the first library call, the client satisfies $\llbracket\sigma_0\rrbracket$;
(ii) The client's calls into the library are exactly those recorded in $r$. The $k$-th call invokes $f_k$ with C arguments $\mathsf{args}_{f_k}(w_k)$ and supplies $\mathsf{BorrowIn}_{f_k}(w_k)$ separably from the resources tracked by $\sigma_{k-1}$; and
(iii) If the $k$-th call returns $\mathsf{ret}$, the resulting state satisfies the body of $Q^E_{\sigma_{k-1},f_k}(w_k,\mathsf{ret})$ in \cref{eq:edge-conditions} for some $o$ and some branch $br$ with $br.\mathsf{next}=\sigma_k$. The client then preserves $\llbracket\sigma_k\rrbracket$ until the next library call.
\end{definition}


Let $\mathsf{Safe}_L(\pi)$ mean that every call into library $L$ during $\pi$ is free of the runtime errors excluded by the underlying program logic.
For \ours's VST instantiation, this scope is described in \Cref{sec:motivation-1-func-verif}.
We now state protocol soundness using \cref{def:protocol-run,def:run-realization} and $\mathsf{Safe}_L(\pi)$, combining the two machine-checked relations in \Cref{fig:assurance-case}:

\emph{(H1) Implementation satisfaction} (\Circled{2}): for every interface function $f\in\mathbb F$, function-level verification (\Cref{sec:motivation-1-func-verif}) establishes that its implementation $C_f$ satisfies $I_f$, under the stated assumptions on external functions.

\emph{(H2) Per-edge subsumption} (\Circled{3}): $\mathsf{Sub}(I_f,E_{\sigma,f})$ holds for every enabled $(\sigma,f)$ pair.

\begin{theorem}[Soundness]\label{thm:protocol-soundness}
Under H1 and H2, every execution $\pi$ that realizes a valid run $r$ is safe:
\begin{equation}\label{eq:protocol-soundness}
  \mathsf{ValidRun}_{\mathcal P}(r)
  \ \land\
  \mathsf{Realizes}_{\mathcal P}(\pi,r)
  \quad\Longrightarrow\quad
  \mathsf{Safe}_L(\pi).
\end{equation}
\end{theorem}

\begin{proof}[Proof sketch]
We induct on the call index $k$ in $r$, maintaining the \emph{invariant}---the client satisfies $\llbracket\sigma_{k-1}\rrbracket$ immediately before the $k$-th call.
The base case is supplied by realization condition (i).

Validity of the current transition supplies the call guard $G_{f_k}(\sigma_{k-1},w_k)$ and, by \Cref{def:protocol-transition}, shows that $(\sigma_{k-1},f_k)$ is enabled.
Together with the invariant, the matching arguments, and the separately supplied borrowed resources $\mathsf{BorrowIn}_{f_k}(w_k)$, these facts establish the edge precondition in \cref{eq:edge-conditions}.
By H1, H2, the frame rule, and the rule of consequence, the implementation satisfies the edge contract.
The call therefore incurs no runtime error covered by the program logic.

If the call returns, realization condition (iii) gives $\llbracket\sigma_k\rrbracket$ and its preservation until the next call, re-establishing the invariant.

For a non-returning call, no successor state is needed: the edge contract's fault-avoidance guarantee covers every step of that call.
Every library call of $\pi$ has a finite index, so the induction covers finite and infinite executions alike.
\end{proof}



\subsection{From Protocol Soundness to Requirements Conformance}\label{sec:cp-quality}

Protocol soundness combines implementation satisfaction (\Circled{2}) and per-edge subsumption (\Circled{3}) in the assurance case of \Cref{fig:assurance-case}; requirements conformance still requires human review (\Circled{4}).
To explain how these relations connect, we use concrete executions to examine what the protocol establishes and what remains to be reviewed, as overviewed in \cref{fig:protocol-execution-scope}.

\begin{figure}[t]
\centering
\includegraphics[width=0.9\linewidth]{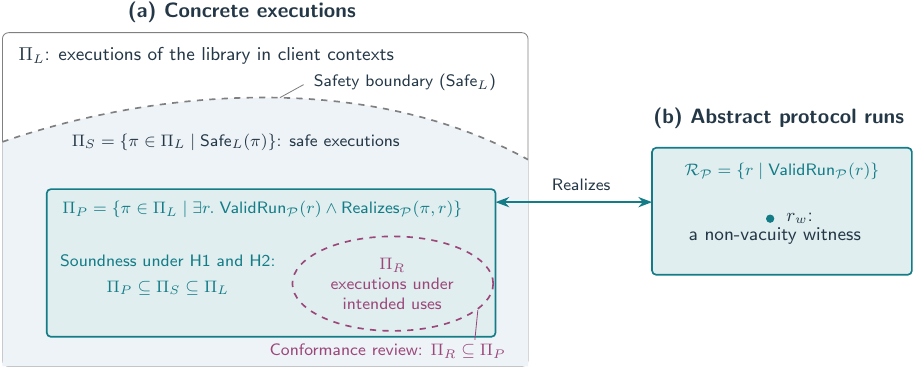}
\caption{A bird’s-eye view of interface-protocol soundness, non-vacuity, and requirements conformance.
}
\label{fig:protocol-execution-scope}
\end{figure}

\paragraph{Executions and soundness.}
Let $\Pi_L$ be the set of all executions of library $L$ in client contexts under the stated environment assumptions; these executions may be unsafe.
The safe executions form $\Pi_S\defeq\{\pi\in\Pi_L\mid\mathsf{Safe}_L(\pi)\}$.
The predicate $\mathsf{Safe}_L$ determines the conceptual safety boundary shown dashed in \cref{fig:protocol-execution-scope}a; characterizing the corresponding safe-input domain, for example through a weakest precondition for a fixed program, is generally nontrivial.

The protocol models executions that each realize at least one valid run, forming $\Pi_P\defeq\{\pi\in\Pi_L\mid\exists r.\;\mathsf{ValidRun}_{\mathcal P}(r)\land\mathsf{Realizes}_{\mathcal P}(\pi,r)\}$.
Under H1 and H2, \Cref{thm:protocol-soundness} establishes $\Pi_P\subseteq\Pi_S\subseteq\Pi_L$.
Under the same hypotheses, every returning call lands in a declared branch, so whether an execution belongs to $\Pi_P$ depends only on how the client uses the library.
Executions outside $\Pi_P$ may still be safe; the protocol is not claimed to capture every safe use.

\paragraph{A non-vacuity witness.}
Let $\mathcal R_{\mathcal P}\defeq\{r\mid\mathsf{ValidRun}_{\mathcal P}(r)\}$ in \cref{fig:protocol-execution-scope}b be the set of valid runs accepted by the protocol.
Let $\Pi_R\subseteq\Pi_L$ denote the executions in which the client uses the library as the requirements intend.
A sound protocol need not be useful: safety alone does not ensure that $\Pi_P\supseteq\Pi_R$.
For example, overly restrictive guards $G_f$ may exclude every run containing an interface call, leaving only zero-call runs; no execution involving a library call can then realize an accepted run.
If all executions in $\Pi_R$ involve interface calls, $\Pi_P\cap\Pi_R=\emptyset$.

To prevent this degeneration, \ours designates a state $\sigma_{\mathrm{init}}$ and proves $\mathsf{Init}_L\vdash\llbracket\sigma_{\mathrm{init}}\rrbracket$,
where $\mathsf{Init}_L$ describes the memory and other logical resources available at library initialization, including explicitly stated assumptions.
The entailment establishes that the chosen $\sigma_{\mathrm{init}}$ is supported by these resources.
\ours then constructs a finite witness run $r_w\in\mathcal R_{\mathcal P}$ from $\sigma_{\mathrm{init}}$, exercising every interface function, and proves that each step's complete edge precondition and selected branch postcondition are satisfiable.
This rules out models that admit no resource-satisfiable use of some interface.
Appendix~\ref{app:covering-witness} illustrates these checks.

\paragraph{Requirements conformance.}
A non-vacuity witness does not establish requirements conformance.
Human review assesses whether the protocol covers intended uses ($\Pi_R\subseteq\Pi_P$) and whether the contracts express the required results.
Discrepancies may call for specification revision or reveal implementation defects.
\Cref{sec:framework-review} describes how \ours supports and records this review.

\section{\ours Framework}\label{sec:framework}

To address \Cref{cha:function-level-verification,cha:library-level-conformance} together, \ours derives contracts from intended uses and cross-call composition requirements, then revises them in response to proof obligations.
The two phases in \Cref{fig:workflow} organize this process around a reviewed account of the intended behavior.
Phase I prepares the source and constructs a natural-language (NL) package of function contracts, cross-call behavior descriptions, and external assumptions.
Human approval of this analysis guides the construction of candidate Rocq specifications and an interface protocol.
Phase II co-evolves these candidates with their proofs, using the approved analysis to justify revisions, and submits the final package for human confirmation.
The inputs comprise a C codebase, a configuration specifying the verification goals and build settings, and requirements supplied through documentation or human review.

In terms of \Cref{fig:assurance-case}, source preparation establishes the translation and artifact correspondence (\Circled{1}), while proof construction establishes implementation satisfaction (\Circled{2}) and per-edge subsumption (\Circled{3}).
The two human checkpoints assess requirements conformance (\Circled{4}) at the analysis and final-evidence stages.
\texttt{REVIEW} can request revised source preparation or analysis, while \texttt{CONFIRM} can return the candidate package to \texttt{PROVE} for revision and rechecking.
The figure shows control flow above and artifact evolution below, omitting intermediate routing and unsuccessful exits.

\begin{figure}[t]
\centering
\includegraphics[width=\columnwidth]{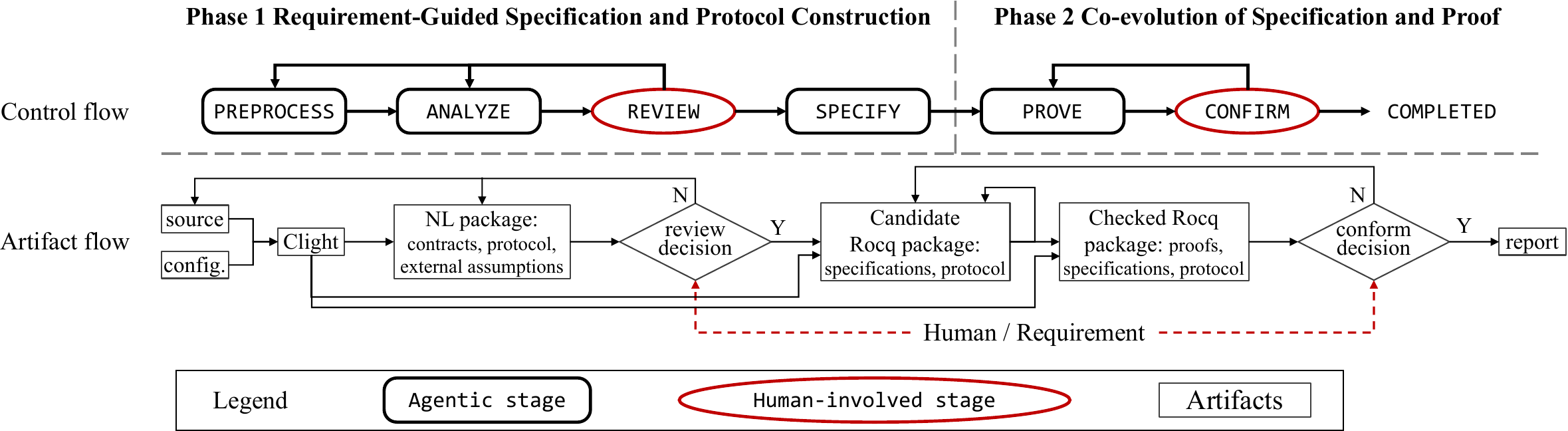}
\caption{\ours's two-phase workflow: control flow above and artifact evolution below.}
\label{fig:workflow}
\end{figure}

\subsection{Phase I: Requirement-Guided Specification and Protocol Construction}\label{sec:framework-phase1}\label{sec:framework-analyze}

\paragraph{Source preparation and translation.}\label{sec:framework-preprocess}
\texttt{PREPROCESS} adapts the codebase for verification: beyond configuring Clight translation, the agent identifies constructs unsupported by VST, such as reachable \texttt{goto}, and selects applicable source transformations.
These transformations must preserve program behavior, including whether execution is defined, rather than repair or conceal library defects.
The changes and their justifications are presented for human review in \texttt{REVIEW}, which can request further source adaptation or analysis.
The source and its corresponding Clight program are finalized before formal specification construction.

\paragraph{Deriving contracts from intended use.}
The weak contracts in \Cref{sec:motivation-2-inter-conform} illustrate why a provable description of an isolated function may omit information its callers need.
In \texttt{ANALYZE}, \ours therefore propagates contractual demands from callers to callees~\cite{fm-agent}.
At each call site, the agent records the argument facts and resources the caller can supply, together with the results needed to continue execution and meet the caller's own contract.
These demands guide the callee's NL contract; forward reasoning through its implementation then checks whether the contract accounts for normal returns, errors, and resource effects.
For an open library, the analysis also considers interface functions across the resource states that client operations can produce, including states left by failed calls.
It records how each operation changes these states and what conditions subsequent calls require.
Candidate preconditions and call-order restrictions are justified against this analysis, rather than used in advance to exclude difficult cases.
The resulting NL package thus describes both individual function behavior and the cross-call relationships that the protocol must capture.

\paragraph{Reviewing the intended uses.}
A precondition suggested by the implementation may be a legitimate caller obligation or conceal a library defect.
To distinguish them, \texttt{REVIEW} presents the proposed preconditions and external assumptions with their justifications, including the uses that each precondition excludes.
Human review assesses whether these conditions agree with the library's intended use and can request revised preparation or analysis.
Approval establishes the NL analysis as the reference for formalization and subsequent revisions.

\paragraph{Formalizing for cross-call composition.}\label{sec:framework-specify}
In \texttt{SPECIFY}, the agent formalizes the approved contracts from callees to callers, so that caller specifications can be checked against available callee contracts.
While the aforementioned top-down analysis establishes contractual demands, this bottom-up order~\cite{cav24-autospec,preguss} organizes formal dependencies.
The agent constructs representation predicates, external and function contracts, and auxiliary lemma statements, checking their intended meaning through NL arguments and counterexample attempts.

For an open library, the agent constructs the formal protocol from the cross-call analysis and candidate interfaces.
States and their interpretations capture resources and facts maintained across calls, while guards and branches describe the conditions and outcomes of each operation.
Constructing edge contracts reveals what the interfaces must preserve for subsequent calls and guides revision of insufficient postconditions.
For example, the insertion and removal contracts in \Cref{sec:motivation-2-inter-conform} must relate the resulting list to the input to justify a later length query.
Together, the specifications and protocol form the candidate Rocq package, whose implementation-satisfaction and subsumption obligations are discharged in Phase II.

\subsection{Phase II: Co-evolution of Specification and Proof}\label{sec:framework-phase2}\label{sec:framework-prove}

\paragraph{Checking implementations and compositions.}
\texttt{PROVE} checks the candidate package through two complementary kinds of proof obligation:
(i) Function-body proofs establish that each implementation satisfies its contract, using representation lemmas and callee contracts.
(ii) For an open library, per-edge subsumption (\cref{def:edge-subsumption}) proofs establish that the interfaces support the protocol's declared behavior under the edge preconditions.
Together, these checks address local verification and cross-interface composition, supplying the two hypotheses of \Cref{thm:protocol-soundness}.
Proof tasks follow the specification dependencies, with independent tasks assigned to parallel agents once their shared contracts are stable.
The final evidence includes the internal callees' body proofs, as well as the protocol's initialization, non-vacuity witness, and run-level obligations for an open library.

\paragraph{Revising candidates from proof feedback.}
As the overflow obligation in \Cref{sec:motivation-1-func-verif} illustrates, a proof attempt can reveal conditions absent from the initial specification.
\ours uses such feedback to co-evolve specifications and proofs (\Cref{cha:function-level-verification}), while checking proposed revisions against the approved intended uses.
The agent distinguishes three cases: (i) missing lemmas or proof strategies call for further proof construction under the existing contract; (ii) suspected contract or representation gaps prompt a reassessment of the NL argument and concrete counterexample attempts before revision; and (iii) confirmed implementation defects are reported rather than hidden by revised contracts.
A newly proposed precondition requires an explanation of why the caller should supply it and which intended uses it excludes; a weaker postcondition must still provide the results required by the approved analysis and subsequent calls.

Revisions can affect both kinds of proof: changing an interface may alter its callers' obligations and the protocol edges that rely on it.
A coordinating agent maintains the shared specifications, records semantic changes, and pauses affected proof tasks while their dependencies are rebuilt.
Direct and transitive dependents are rechecked after changes to specifications, representation predicates, or the protocol.
This coordination allows parallel proof construction while ensuring that the resulting evidence refers to one consistent candidate package.

\paragraph{Closing the assurance case.}\label{sec:framework-review}
Before \texttt{CONFIRM}, \ours rebuilds the proofs and checks their coverage, absence of admitted obligations, and correspondence with the frozen source and Clight program.
It also audits the axioms and external premises used by the proofs.

In \texttt{CONFIRM}, reviewers compare the final contracts and protocol with the approved NL analysis and recorded revisions, checking intended uses, required results, and external assumptions.
Rejected candidates return to \texttt{PROVE} for revision and rechecking.
\Cref{sec:case-study} illustrates how reviewing a calling restriction can expose an implementation defect.
On approval, \ours delivers the checked package with a report of its assumptions and human confirmation.
If verification cannot be completed, the report identifies the unresolved obstacles and unverified scope.

\section{Evaluation}\label{sec:evaluation}\label{sec:evidence}

We organize the evaluation around the following research questions:
\begin{itemize}
  \item[\textbf{RQ1:}] Can \ours construct assurance cases for conformance of C codebases to requirements for memory safety, interface composition, and functional correctness?
  \item[\textbf{RQ2:}] What are the performance bottlenecks in \ours's verification process?
  \item[\textbf{RQ3:}] How does \ours compare against general-purpose coding agents?
\end{itemize}
We first report the evaluation setup (\cref{sec:evaluation-setup}), then address \textbf{RQ1}, \textbf{RQ2}, and \textbf{RQ3} (\cref{sec:RQ1,sec:RQ2,sec:scaffolding-ladder}) and conduct a case study (\cref{sec:case-study}), and finally discuss threats to validity (\cref{sec:threats-to-validity}).

\subsection{Evaluation Setup}\label{sec:evaluation-setup}

\paragraph{Implementation.}
\ours comprises approximately 17,300 lines of Python and shell code for scaffolding, verification gates, and evidence generation, a 1,500-line OCaml Clight inspector, and 6,600 lines of YAML workflow definitions and proof guidance.
It runs on the FSM engine \cite{mcp-server-fsm} and integrates an OpenCode development build based on upstream version~1.18.31 with Rocq/VST and CompCert.

\paragraph{Configurations.}
We use multiple models and experimental machines because the studies span an extended development period during which model procurement, service availability, and access to computing resources changed.
The principal models are GLM-5.3-Flash (glm-F), GLM-5.3 (glm), and GPT-6-Luna (gpt-L).
Appendix~\ref{app:run-configurations} summarizes the four experimental environments; RQ3 records host, model, toolchain, and resource differences explicitly rather than assuming identical configurations.


\paragraph{Benchmarks.}
We evaluate three groups of C codebases, with their statistics and selected properties summarized in \cref{tab:eval-benchmarks}.
The evaluation comprises six benchmark families and seven implementation cases, counting the two red--black tree implementations separately; the additional SHA-256/224 model comparison yields eight rows in the table.
Each row records an evaluated configuration, which may aggregate multiple execution segments; the Total row excludes the additional SHA run.
(i) \emph{Small academic cases of classical data structures} comprise a linked-list library (LL) and red--black tree implementations (RBT), exercising allocation, traversal, ownership transfer, parent pointers, and rotations.
(ii) \emph{Previously verified industrial modules} are represented by the AES-128/CCM* code from Contiki (CTK), an open-source operating system for the Internet of Things~\cite{contiki-source}, with prior deductive-verification work~\cite{contiki}.
(iii) \emph{New industrial verification targets} come from openHiTLS~\cite{openhitls-source}: SHA-256/224 hashing (SHA), Classic McEliece key encapsulation (MCE), and the post-quantum signature implementation ML-DSA (ML).
MCE covers key generation, encapsulation/decapsulation, control-bit construction, and polynomial operations; the ML-DSA functional case targets the core and NTT modules for parameter set 65.

\paragraph{Verified artifacts.}
The reported guarantees hold for each benchmark as verified, which relates to the original code in three ways.
First, each benchmark is verified at a fixed source revision and under one build configuration: the x86-64 target and one set of preprocessor macros, both recorded in the supplementary package.
The openHiTLS cases use the library's default profile with its portable C implementations; for Classic McEliece we additionally select the upstream alignment-enforcing configuration (\texttt{FORCE\_ADDR\_ALIGN}).
Second, where VST cannot express a construct directly, the source undergoes a behavior-preserving adaptation, such as inlining a cleanup jump or naming an anonymous type.
The kinds are few, and each adaptation is recorded with its justification and reviewed.
Third, the guarantees cover inputs that satisfy the interface preconditions and, where an interface protocol is constructed, call sequences that the protocol admits.

\subsection{RQ1: Construction of Assurance Cases}\label{sec:RQ1}

\ours constructed assurance cases connecting the benchmark implementations to their selected requirements: memory safety in all cases, interface composition for the five cases with an interface protocol, and additional functional correctness for the red--black tree and ML-DSA cases (\cref{tab:eval-benchmarks}).
Human conformance review confirmed that the contracts and, where applicable, the interface protocols capture the intended requirements under the reviewed input and usage conditions.
The cases cover 259 verified functions and 6,626 lines of C (68,886 lines of generated Clight), with 500,617 lines of Rocq code, an aggregate time upper bound of 2,694.58 hours, 209.803 billion recorded tokens, and 25.65 hours of estimated human review.
These results demonstrate \ours's high degree of automation in proof construction, which typically requires substantial specialist effort. For example, Apple's verification of its own ML-KEM and ML-DSA implementations (distinct from our benchmarks) involved more than 50,000 proof steps and development of reusable Isabelle proof libraries~\cite{mldsa-apple}.

\paragraph{Memory safety.}
All functions selected for verification satisfy their contracts through function-level verification, and the delivered Rocq proofs contain no \texttt{Admitted} obligations.
The evidence binds these proofs to the verified source configuration and records the external assumptions on which they depend.
These proofs establish memory safety for calls satisfying the top-level interface preconditions and disclosed external assumptions.

\begin{table}[t]
  \centering
  \footnotesize
  \begin{threeparttable}
  \caption{Effectiveness of \ours in constructing assurance cases for C codebases and requirements, with verified properties, function coverage, and recorded costs.
  Props. lists the selected properties: memory safety (M), functional correctness (C), and interface composition through a protocol (P); Model records the LLM configuration.
  C/PP/AST/Rocq count physical lines of original C implementation files, final preprocessed files, generated Clight, and delivered Rocq (specifications, proofs, and supporting code excluding Clight), respectively.
  \#Fn and \#API count function definitions and protocol-covered interfaces; --- denotes an inapplicable protocol.
  Proofs gives proved/total function bodies, not functional-goal coverage.
  Time is net verification time in hours after deducting documented review waits and stoppages, except where marked; Tok. is recorded model token usage including cache ($\mathrm{B}=10^9$).
  Human is estimated review effort in hours; Missing denotes unavailable token measurements.
  Total omits the duplicate SHA/GLM-5.3 run but retains both RBT cases with distinct implementations and properties.}
  \label{tab:eval-benchmarks}
  \renewcommand{\arraystretch}{0.90}
  \setlength{\tabcolsep}{3.65pt}
  \begin{tabular}{@{}lllrrrrrrrrrr@{}}
    \toprule
    & \multicolumn{2}{c}{Settings} & \multicolumn{5}{c}{Benchmark statistics} & \multicolumn{5}{c}{Verification results} \\
    \cmidrule(lr){2-3}\cmidrule(lr){4-8}\cmidrule(l){9-13}
    Bench. & Props. & Model & C & PP & AST & \#Fn & \#API & Proofs & Rocq & Time (h) & Tok. (B) & Human (h) \\
    \midrule
    LL  & M, P & glm-F & 135 & 449 & 1,601 & 10 & 10 & 10/10 & 9,565 & 21.42 & 1.526 & 0.17 \\
    RBT & M & glm-F & 287 & 331 & 2,112 & 12 & --- & 12/12 & 17,531 & 157.16 & 7.302 & 3.62 \\
    CTK & M, P & glm-F & 360 & 527 & 2,522 & 10 & 1 & 10/10 & 27,857 & 259.45 & 16.138 & 0.57 \\
    MCE & M, P & glm,glm-F & 2,698 & 19,787 & 25,012 & 138 & 15 & 138/138 & 229,025 & 965.87\tnote{*} & 106.775 & 4.00 \\
    SHA & M, P & glm-F & 1,418 & 7,739 & 7,337 & 25 & 21 & 25/25 & 44,391 & 214.58 & 14.994 & 7.29 \\
    SHA & M, P & glm & 1,418 & 7,739 & 7,337 & 25 & 21 & 25/25 & 58,488 & 148.19 & 10.991 & 7.99 \\
    ML & M, C & Mix\tnote{\textdagger} & 1,502 & 8,063 & 28,781 & 55 & --- & 55/55 & 142,152 & 668.1\tnote{*} & Missing\tnote{\textdaggerdbl} & 8.00 \\
    RBT & M, C, P & gpt-L,glm-F & 226 & 200 & 1,521 & 9 & 5 & 9/9 & 30,096 & 408.00\tnote{*} & 63.068 & 2.00 \\
    \midrule
    Total & --- & --- & 6,626 & 37,096 & 68,886 & 259 & 52 & 259/259 & 500,617 & 2,694.58\tnote{*} & 209.803\tnote{\textdaggerdbl} & 25.65 \\
    \bottomrule
  \end{tabular}
  \begin{tablenotes}[flushleft]
    \scriptsize
    \item[*] The time is an upper-bound estimate that includes unquantified review waits, pauses, and recovery intervals.
    \item[\textdagger] Mix: Claude Opus 4.8, GLM-5.2, Kimi K3, DeepSeek V4 Flash, GLM-5.3, and GLM-5.3-Flash.
    \item[\textdaggerdbl] ML-DSA token records were lost over an extended development period spanning multiple model iterations; Total excludes its token usage.
  \end{tablenotes}
  \end{threeparttable}
\end{table}

\paragraph{Library-call composition.}
All benchmark runs except the red--black tree run with GLM-5.3-Flash and the ML-DSA case construct an interface protocol. The former supplies \texttt{main} as a designated client and is verified through its entry-point contract; the latter is an exploration of functional correctness.
As illustrated in \cref{sec:motivation-2-inter-conform}, specifications sufficient for individual function proofs may be too weak to support the intended interface-call sequences; constructing the protocol exposes the additional guarantees that the specifications must establish and preserve.
For example, McEliece key generation can safely return a failure while retaining allocated but incompletely initialized key storage; safe cleanup requires ownership of that storage, not a usable key.
A subsequent encapsulation call, however, needs initialized public-key material, which the contracts must provide and retain along the intended call sequence (Appendix~\ref{app:mceliece-protocol}).
Constructing and reviewing the contracts needed for protocol composition can also expose implementation defects, as illustrated in \cref{sec:case-study}.

\paragraph{Functional correctness.}
\ours establishes functional correctness under specified contracts for the following cases. Red--black trees are widely used self-balancing search structures.
For the second red--black tree implementation in \cref{tab:eval-benchmarks}, a textbook C implementation written by the authors, \ours constructs Rocq-checked proofs of the following guarantees.
(i) \emph{Search ordering and red--black invariants}: strict search ordering, legal colors, a black root, no adjacent red nodes, and equal black height.
(ii) \emph{Correct operations}: successful insertion adds the specified key while preserving the invariants, whereas failed insertion preserves the original tree; search preserves the tree and returns the corresponding node exactly when the key is present; creation and destruction satisfy their initialization and resource contracts.
The interface protocol connects these operation contracts across permitted calls, preserving the tree invariants for subsequent operations and ruling out insertion or search after partial destruction.
Proving such properties directly on C source is known to be hard, because the correctness argument is intertwined with reasoning about pointer aliasing and side effects~\cite{schellhorn-rbtree-vstte22}; here the proofs are constructed by the agent and checked by Rocq.
All the details are provided in Appendix~\ref{app:rbtree-functional}.

ML-DSA is an exploratory case: its functional proofs are established per function and rely on the premises listed in Appendix~\ref{app:mldsa-functional}.
\ours establishes three functional guarantees for the internal procedures against the mathematical model of FIPS~204 Algorithms~6--8~\cite{fips204}, under their specified preconditions.
(i) \emph{Key generation}: on success, the public- and private-key buffers contain the keys computed from the input seed.
(ii) \emph{Signing}: on success, the signature bytes and reported length match the mathematical signing result for the supplied inputs.
(iii) \emph{Verification}: a successful C return implies that the mathematical verifier accepts the supplied signature.
The functional guarantees above are established by refinement proofs.
These are partial-correctness guarantees subject to the reviewed sampling, nonce-range, and external SHAKE assumptions; the detailed contracts, proof connections, and assumptions are provided in Appendix~\ref{app:mldsa-functional}.

\paragraph{Model variation.}\label{sec:model-variation}
Both the GLM-5.3-Flash and GLM-5.3 SHA-256/224 runs completed the verification task, providing evidence that \ours is not tied to a specific model.
Model choice nevertheless affected the observed verification effort:
GLM-5.3 produced more precise NL contracts that required greater review effort, yet its proofs involved fewer failed attempts and revisions than those of GLM-5.3-Flash, resulting in lower net verification time and token usage.


\subsection{RQ2: Cost Distribution and Bottlenecks}\label{sec:RQ2}

Despite this high degree of automation, \ours still incurs substantial machine time and token usage (\cref{tab:eval-benchmarks}).
We investigate the main bottlenecks by examining how these costs are distributed across verification stages, execution activities, and function-proof tasks.

\paragraph{Cost by verification stage.}
\texttt{PROVE} dominates verification cost, accounting for over 82\% of machine time and 88\% of tokens in each measured benchmark, normalized over the four automated stages (\cref{fig:rq2-stages}); it is therefore the main target for cost reduction.
The following two analyses therefore focus on \texttt{PROVE}: its execution activities and the distribution of costs across function-proof tasks.

\begin{figure}[t]
  \centering
  \captionsetup[subfigure]{skip=3pt} 
  \begin{subfigure}[t]{0.49\linewidth}
    \centering
    \vspace{0pt}
    \includegraphics[width=\linewidth]{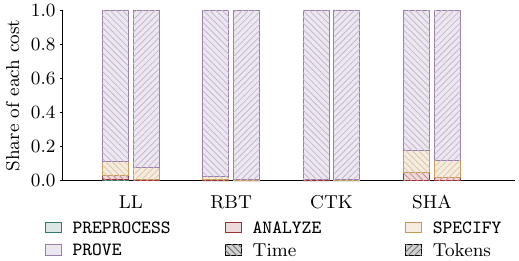}
    \caption{Time and token usage by agentic stage in \cref{fig:workflow}.}
    \label{fig:rq2-stages}
  \end{subfigure}\hfill
  \begin{subfigure}[t]{0.49\linewidth}
    \centering
    \vspace{0pt}
    \includegraphics[width=\linewidth]{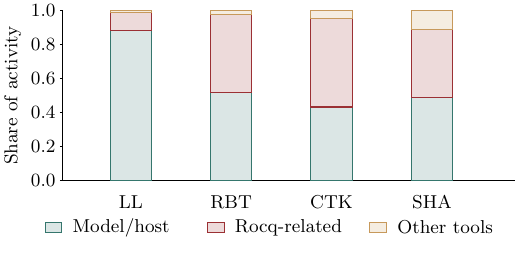}
    \caption{Recorded activity time in \texttt{PROVE}.}
    \label{fig:rq2-activities}
  \end{subfigure}
  \par\medskip
  \begin{subfigure}[t]{\linewidth}
    \centering
    \includegraphics[width=\linewidth]{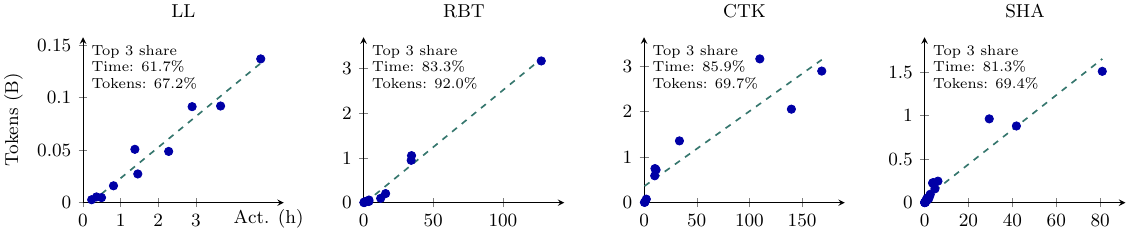}
    \caption{Per-function proof-task costs: worker activity time (Act. (h)) and tokens (B). Annotations give the top-3 share of total function activity time and tokens. Dashed lines are descriptive least-squares linear fittings.}
    \label{fig:rq2-functions}
  \end{subfigure}
  \caption{Cost distributions across benchmarks with available measurements.
  }
  \label{fig:rq2-cost-breakdown}\label{fig:stage-tokens}
  \label{fig:rq2-rbtree}\label{fig:rq2-contiki}
\end{figure}

\paragraph{Model and tool-call time within \texttt{PROVE}}
An agent's recorded working time comprises model calls and tool calls; token usage is incurred by model calls.
Rocq-related calls dominate tool use, accounting for 78.7\%--95.1\% of classified tool-call time across the four benchmarks, while other tools contribute less than 11\% of total classified activity time (\cref{fig:rq2-activities}).
During proof construction, agents repeatedly revise proof scripts and invoke Rocq to check the changes; recompiling modified files and affected dependencies adds to the cost of producing the final proof.
Reducing this repeated compilation and checking is therefore a primary optimization direction (\cref{sec:limitations}).

\paragraph{Concentration in function-proof tasks.}
\texttt{PROVE} encompasses function-body proofs, interface protocols, auxiliary lemmas, and representation predicates.
Within the function-proof tasks shown in \cref{fig:rq2-functions}, the top three functions account for at least 61.7\% of activity time and 67.2\% of token usage in each benchmark.
The SHA \texttt{CompressBlock} proof illustrates how proof size can contribute to this \emph{concentration}: even after its tactics completed, the final \texttt{Qed} kernel check exceeded the command timeout.
Larger proof scripts can also require more rounds of revision, compounding the cost as each attempt incurs further model calls and repeats expensive Rocq checks.
The challenge is therefore not only to find a valid proof, but also to construct one that can be revised and checked efficiently.
We discuss possible improvements to proof decomposition in \cref{sec:limitations}.

\begin{table}[t]
  \centering
  \small
  \begin{threeparttable}
  \caption{Agent baselines and \ours on SHA-256/224 using GLM-5.3-Flash.}
  \label{tab:ladder}
  \renewcommand{\arraystretch}{0.82}
  \setlength{\aboverulesep}{0.3ex}
  \setlength{\belowrulesep}{0.4ex}
  \setlength{\tabcolsep}{3pt}
  \begin{tabular*}{\linewidth}{@{\extracolsep{\fill}}lrrrrrl@{}}
    \toprule
    Method & Proofs & Rocq & Time (h) & Tok. (B) & Human (h) & Outcome \\
    \midrule
    \textsc{OC}\textsubscript{bare} & 11/25 & 859 & 5.96 & 0.244 & --- & Incomplete \\
    \textsc{OC}\textsubscript{goal} & 19/25 & 2,199 & 99.68 & 1.902 & --- & Incomplete \\
    \ours & 25/25 & 44,391 & 214.58 & 14.994 & 7.29 & Delivered \\
    \bottomrule
  \end{tabular*}
  \end{threeparttable}
\end{table}

\subsection{RQ3: Comparison with Agent Baselines}\label{sec:scaffolding-ladder}\label{sec:RQ3}

We examine what \ours's verification workflow adds beyond general-purpose coding agents by comparing \textsc{OC}\textsubscript{bare}, \textsc{OC}\textsubscript{goal}, and \ours on SHA-256/224 with GLM-5.3-Flash (\cref{tab:ladder}).
\textsc{OC}\textsubscript{bare} uses the same development build of OpenCode as \ours, with the sources, Rocq/VST toolchain, and a task statement, but without \ours's verification workflow, serving as an ablation.
\textsc{OC}\textsubscript{goal} uses OpenCode~2.2.1 for its native \texttt{/goal} mechanism~\cite{opencode-native-goal-release}, which maintains a persistent objective and automatically continues execution until completion or a stopping condition.
This mechanism is unavailable in the development build used by \ours.

On SHA-256/224, \ours enabled GLM-5.3-Flash to complete the verification task, whereas \textsc{OC}\textsubscript{bare} remained incomplete; enabling \texttt{/goal} yielded more proof progress but still fell short of completion.
\textsc{OC}\textsubscript{bare} stopped on its own after proving 11 of 25 bodies, whereas \textsc{OC}\textsubscript{goal} reached 19 before automatically pausing after prolonged waiting for background work; \ours completed all 25 with human conformance review.
Beyond incomplete coverage, the baseline artifacts exhibited specification defects and gaps in composition and build evidence, so their local proof counts did not establish library-wide verification.
Appendix~\ref{app:agent-baselines} provides the detailed analysis.


\subsection{Case Study}\label{sec:case-study}

\begin{figure}[t]
\centering
\begin{subfigure}[t]{0.51\columnwidth}
\centering
\vspace{0pt}
\fbox{\ttfamily\scriptsize
\begin{tabular}{@{}l@{}}
  \textbf{typedef struct} \{ \textbf{uint32\_t} algId; \textbf{uint32\_t} paraN;\\
  \ \ \ \ \textbf{uint8\_t} *randBuf; \textbf{uint32\_t} randBufLen; \} \textbf{Ctx};\\[0.2ex]
  \textbf{Ctx} *NewCtx(\textbf{void}) \{\ /* Initialize a fresh Ctx */\ \}\\
  \textbf{int} SetRandBuf(\textbf{Ctx} *c, \textbf{void} *v, \textbf{uint32\_t} len) \{\\
  \ \ \textbf{if} (v == NULL || \colorbox{green!24}{len != c->paraN}) \textbf{return} -1;\ /* L2 */\\
  \ \ free(c->randBuf); c->randBuf = NULL;\\
  \ \ \colorbox{green!24}{c->randBuf = malloc(len);}\ /* L3 */\\
  \ \ \textbf{if} (c->randBuf == NULL) \textbf{return} -1;\\
  \ \ {\setlength{\fboxsep}{0pt}\colorbox{green!24}{memcpy(c->randBuf, v, len); c->randBufLen = len;}}\ /* L4 */\\
  \ \ \textbf{return} 0;\\
  \}\\[0.2ex]
  \textbf{int} SetAlgId(\textbf{Ctx} *c, \textbf{int32\_t} id) \{\\
  \ \ \textbf{if} (c->algId != 0) \textbf{return} -1;\\
  \ \ {\setlength{\fboxsep}{0pt}\colorbox{green!24}{c->algId = id; c->paraN = 16;}}\ /* L1 */\\
  \ \ \textbf{return} 0;\\
  \}\\[0.2ex]
  \textbf{int} GetRandBuf(\textbf{Ctx} *c) \{\\
  \ \ \textbf{uint8\_t} tmp[32];\\
  \ \ \textbf{if} (c->randBuf != NULL)\\
  \ \ \ \ {\setlength{\fboxsep}{0pt}\colorbox{orange!24}{memcpy(tmp, c->randBuf, c->paraN);\ /* L5: OOB risk */}}\\
  \ \ ... /* other paths omitted */\\
  \}
\end{tabular}}
\caption{The simplified interface functions.}
\label{fig:case-randBuf-functions}
\end{subfigure}
\hfill
\begin{minipage}[t]{0.48\columnwidth}
\vspace{0pt}
\begin{subfigure}[t]{\linewidth}
\centering
\vspace{0pt}
\fbox{\ttfamily\scriptsize
\begin{tabular}{@{}l@{}}
  \textbf{Definition} SetRandBuf\_spec :=\\
  \ \ DECLARE \_SetRandBuf\\
  \ \ WITH c: val, v: val, len: Z, s: CtxState, ...\\
  \ \ PRE [ tptr t\_Ctx, ... ]\\
  \ \ \ \ PROP ({\setlength{\fboxsep}{0pt}\colorbox{red!18}{algId s <> 0}}; ...)\\
  \ \ \ \ PARAMS (c; v; Vint (Int.repr len))\\
  \ \ \ \ GLOBALS (...) SEP (ctx\_rep c s; ...)\\
  \ \ POST [ tint ] ... .
\end{tabular}}
\caption{The simplified \texttt{SetRandBuf} specification.}
\label{fig:case-randBuf-spec}
\end{subfigure}
\par\medskip
\begin{subfigure}[t]{\linewidth}
\centering
\includegraphics[width=0.89\linewidth]{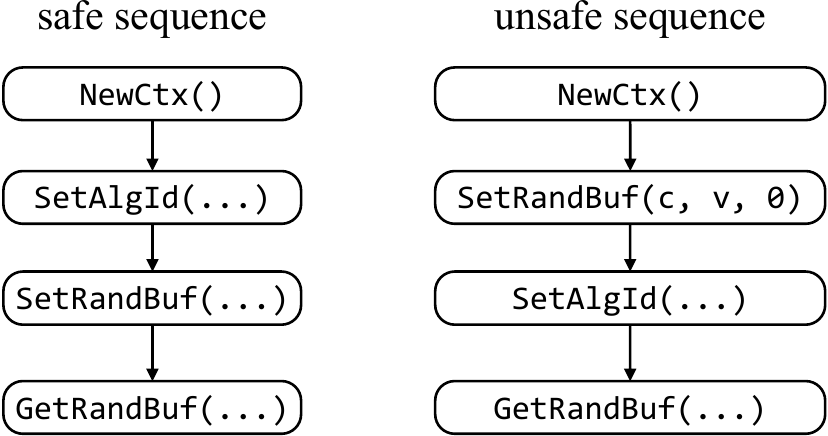}
\caption{The safe and unsafe sequences.}
\label{fig:case-randBuf-orders}
\end{subfigure}
\end{minipage}
\caption{A defect in SLH-DSA: \ours synthesizes a {\setlength{\fboxsep}{0pt}\colorbox{red!18}{strong precondition}} (b) that excludes the unsafe call sequence (c), which triggers an {\setlength{\fboxsep}{0pt}\colorbox{orange!24}{out-of-bounds access}} in (a).}
\label{fig:case-randBuf}
\end{figure}



We use a historical defect in SLH-DSA\footnote{The historical implementation did not satisfy the intended requirements and is therefore excluded from RQ1 benchmarks.} to illustrate how \ours helps locate a mismatch between the implementation and its intended design.
\Cref{fig:case-randBuf} shows the minimal interface implementations and relevant specification, together with two contrasting call sequences: one admitted by the generated precondition and one excluded by it.

In the safe sequence in \cref{fig:case-randBuf-orders}, \texttt{NewCtx} creates a zero-initialized context.
\texttt{SetAlgId} sets \texttt{algId}, the algorithm id that selects the parameters, and initializes \texttt{paraN} to their required random-input length (\texttt{L1}).
\texttt{SetRandBuf} checks \texttt{len == paraN} (\texttt{L2}), allocates a buffer of that length (\texttt{L3}), and records the length in \texttt{randBufLen} (\texttt{L4}); on a successful return, \texttt{randBufLen == paraN} holds.
\texttt{GetRandBuf} then copies \texttt{paraN} bytes from the buffer (\texttt{L5}).
This access is safe here because the preceding call established \texttt{randBufLen == paraN}; a shorter buffer would make \texttt{L5} read out of bounds (OOB).

The generated \texttt{SetRandBuf} specification (\cref{fig:case-randBuf-spec}) highlights the precondition \colorbox{red!18}{\texttt{algId s <> 0}}, i.e., \texttt{SetAlgId} must precede \texttt{SetRandBuf}.
This condition is unnecessary for proving \texttt{SetRandBuf} itself.
It excludes the unsafe sequence in \cref{fig:case-randBuf-orders}.
On a fresh context, \texttt{paraN} is zero, so \texttt{SetRandBuf(c, v, 0)} passes the length check and installs a zero-length buffer (if \texttt{malloc(0)} succeeds).
\texttt{SetAlgId} then sets \texttt{paraN} to 16 without updating the buffer, causing \texttt{GetRandBuf} to read beyond its allocation.

Safe buffer reuse therefore depends on this order.
Making the requirement explicit guided the reviewer to a call sequence that the implementation accepts but cannot safely handle.
After \ours located this defect, the upstream commit~\cite{revision-commit} eliminated the stale-length path by replacing the persistent buffer with a buffer local to each call to \texttt{GetRandBuf}.

\subsection{Threats to Validity}\label{sec:threats-to-validity}

\paragraph{Experimental configuration.}
The experiments span an extended development period, and some cases were verified while the framework itself was evolving.
Rocq's memory demands also constrained parallel execution: contention among concurrent experiments and proof workers could prolong compilation and trigger retries.
We distributed experiments across multiple machines to alleviate this pressure, but the reported times remain affected by resource availability and should not be interpreted as performance under dedicated, adequately provisioned hardware.
The high cost precluded repeated runs, leaving LLM-induced variability unquantified.



\paragraph{Human review.} The experiments were conducted and reviewed by engineers and doctoral students with varying verification experience. The human hours in \cref{tab:eval-benchmarks} thus reflect reviewer expertise as well as task demands, limiting direct comparisons of effort across cases. Conformance review may also miss specification omissions or unjustified restrictions.
\section{Limitations and Future Work}\label{sec:limitations}
We pinpoint scenarios where \ours is inadequate and provide potential solutions thereof.

Repeated Rocq checks remain costly during proof construction (RQ2).
Splitting large proof files into smaller units and reusing checked results for unchanged units could avoid replaying unaffected proof scripts after local edits, shortening the feedback cycle.

Large function proofs can also be expensive to check even after proof search succeeds, as the SHA \texttt{CompressBlock} case in RQ2 illustrates.
Decomposing these proofs earlier into independently checked lemmas could keep proof terms manageable and localize subsequent revisions.
This would speed up individual proofs within the co-evolution process of \cref{sec:framework-phase2}.

VST's current proof rules do not support complete C features, such as \texttt{goto}, restricting direct verification of C code with cleanup and error-handling jumps.
A promising extension is to validate the transformation of these features at the Clight level: a transformation would produce structured code and a semantic-preservation witness checked by a verified checker.
Starting with common cleanup patterns could broaden coverage while replacing human-reviewed transformation justifications with machine-checked evidence.
\section{Related Work}\label{sec:related-work}

\subsection{Resource-Aware API Models}
\label{sec:related-resource-protocols}

\paragraph{Resource contracts and modular composition.}
VST specifies C functions with separation-logic contracts~\cite{VST-esop11}; Beringer's Verified Software Units compose verified modules through abstract interfaces~\cite{VSU-esop21}.
RefinedC combines ownership and refinement types with foundational separation-logic proofs~\cite{RefinedC-pldi21}.
\ours builds on this expressiveness: its interface protocol makes the intended cross-call compositions explicit, including borrowed resources and outcome-dependent successors.
Per-edge subsumption requires exported contracts to support these compositions, exposing missing guarantees even when individual contracts are provable.
Human review assesses whether these contracts capture the intended requirements.

\paragraph{Legal-call protocols and interface synthesis.}
Typestate and Vault track how operations change resource states and restrict subsequent uses~\cite{typestate-tse86,Vault-pldi01}.
Henzinger et al. synthesize safe and permissive call interfaces through abstract reachability and refinement~\cite{permissive-interface-fse05}; Alur et al. combine predicate abstraction, automata learning, and model checking~\cite{Java-interface-popl05}.
HATs, implemented in Marple, use symbolic finite automata and refinement types to verify data-dependent interaction histories with stateful libraries~\cite{hats-pldi24}.
\ours establishes the implementation-to-protocol connection through two layers of machine-checked evidence.
First, VST body proofs establish that the Clight implementations satisfy their exported contracts.
Second, VST funspec subsumption and a Rocq closure theorem establish that these same contracts support the protocol's edge obligations.
Human review assesses its conformance to the intended requirements, while establishing that a particular client follows the protocol remains a separate task.

\subsection{Automating Deductive Verification}
\label{sec:related-automation}

\paragraph{Proving supplied specifications.}
CoqHammer reconstructs automated proofs in Coq~\cite{coqhammer-jar18}, while Proverbot9001, PALM, and Rango learn or generate proofs of supplied statements~\cite{proverbot9001,palm-ase24,rango-icse25}.
AutoRocq targets program-verification lemmas extracted by Frama-C/WP from annotated C; its agent constructs Rocq proofs of these obligations~\cite{autorocq}.
In auto-active verification, Dafny and Verus discharge verification conditions through SMT solving~\cite{dafny,verus-sosp24}; Laurel and AutoVerus generate assertions or proof annotations to help this process~\cite{laurel-oopsla25,autoverus-oopsla25}.
RefinedC automates foundational proofs of C code from user-supplied refinement and ownership types~\cite{RefinedC-pldi21}.
In \ours, failed body proofs can trigger contract and resource-abstraction revisions, assessed against recorded requirements.

\paragraph{Synthesizing specifications.}
AutoSpec and Preguss generate ACSL annotations and validate them through deductive verification~\cite{cav24-autospec,preguss}.
Verdict reconciles assumed callee contracts through top-down caller demands and bottom-up updates using Frama-C/WP~\cite{verdict}.
VeriStruct generates abstractions, invariants, contracts, and proof annotations for Rust data-structure modules, guided by test clients and checked by Verus~\cite{veristruct-tacas26}.
Dualis infers client-sufficient library contracts validated by testing~\cite{dualis-pldi26}.
Quiver infers separation-logic specifications from user sketches while establishing their correctness in Coq~\cite{quiver-pldi24}.
\ours combines caller demands with resource-protocol compositions to shape contracts, coupling their revision with foundational body and protocol proofs.
Its top-down analysis draws on FM-Agent's natural-language Hoare-style reasoning~\cite{fm-agent}; \ours additionally formalizes these demands and checks the resulting proofs with Rocq.
The resulting proof establishes implementation satisfaction, while explicit human review assesses whether the final contracts and protocol still express the intended requirements.

\paragraph{Agent-driven verification at scale.}
Two recent efforts use LLM agents to construct machine-checked proofs about compiled code.
AlphaProof Code~\cite{alphaproof-code} proves protocol-level security properties of an executable built from the unmodified \texttt{libsignal} implementation of the Signal protocol, at the AArch64 machine-code level in Lean; an agent drafts the specifications under human guidance, and the user approves them before proving starts.
Anthropic and Theorem~\cite{anthropic-formal-verification} prove, at the RISC-V machine-code level in Rocq, properties of 239 fuzz-test entry points across 123 OSS-Fuzz libraries.
Each entry point fixes how the library is called, and the property is that the test exits normally on every input below 64\,MB, without failed assertions or crashes reported by AddressSanitizer.
Like \ours, both consider single-threaded executions.
\ours differs in its other choices: it works on C source under the Clight semantics, constructs a contract for every function, and, for a library without a designated client, constructs an interface protocol over the permitted call sequences.

\section{Conclusion}\label{sec:conclusion}

We have presented \ours, an LLM-assisted framework for building auditable assurance cases that connect machine-checked program proofs to the intended requirements.
Its two-phase method co-evolves function specifications and proofs, while an interface protocol models intended cross-interface use and guides the construction of composable contracts.
Using VST in Rocq, \ours establishes memory safety under disclosed contracts and assumptions for all 259 selected function definitions across seven implementation cases from six benchmark families, including industrial cryptographic modules.
The red--black tree and ML-DSA cases additionally establish their specified functional properties.
Estimated human review effort ranges from 0.17 to 8.00 hours per reported case.
These cases combine automated evidence construction with human conformance review.

\paragraph{AI Use.}
The reported models, through harnesses such as OpenCode, generate and revise specifications, protocols, and Rocq proofs.
\section*{Data-Availability Statement}\label{sec:data-availability}

An anonymized supplementary package accompanies this submission. It contains the Rocq proof developments supporting the reported results, with instructions for independent proof checking.
We intend to submit the artifact for Artifact Evaluation.

\bibliography{references}

@String{Computing = "Computing" }

@String{Computer = "{IEEE} Computer" }

@String{Springer = "Springer-Verlag" }

@online{doclicense,
  author =    {Robin Schneider},
  title =  {The \textsl{doclicense} package},
  year = 2022,
  url =    {http://www.ctan.org/pkg/doclicense},
  lastaccessed = {May 27, 2022}
  }

@ArtifactSoftware{R,
    title = {R: A Language and Environment for Statistical Computing},
    author = {{R Core Team}},
    organization = {R Foundation for Statistical Computing},
    address = {Vienna, Austria},
    year = {2019},
    url = {https://www.R-project.org/},
}

@book{c-standard,
  title        = {ISO/IEC 9899:2024 Information technology — Programming languages — C},
  author       = {{ISO/IEC JTC 1/SC 22}},
  year         = {2024},
  publisher    = {{ISO}},
  url          = {https://www.iso.org/standard/82075.html}
}

@inproceedings{dafny,
  author       = {K. Rustan M. Leino},
  title        = {Dafny: An Automatic Program Verifier for Functional Correctness},
  booktitle    = {{LPAR} (Dakar)},
  series       = {Lecture Notes in Computer Science},
  volume       = {6355},
  pages        = {348--370},
  publisher    = {Springer},
  year         = {2010},
  doi = {10.1007/978-3-642-17511-4_20}
}

@inproceedings{cav24-autospec,
  author       = {Cheng Wen and
                  Jialun Cao and
                  Jie Su and
                  Zhiwu Xu and
                  Shengchao Qin and
                  Mengda He and
                  Haokun Li and
                  Shing{-}Chi Cheung and
                  Cong Tian},
  title        = {Enchanting Program Specification Synthesis by Large Language Models
                  Using Static Analysis and Program Verification},
  booktitle    = {{CAV} {(2)}},
  series       = {Lecture Notes in Computer Science},
  volume       = {14682},
  pages        = {302--328},
  publisher    = {Springer},
  year         = {2024},
  doi = {10.1007/978-3-031-65630-9_16}
}

@inproceedings{icse25-specgen,
  author       = {Lezhi Ma and
                  Shangqing Liu and
                  Yi Li and
                  Xiaofei Xie and
                  Lei Bu},
  title        = {SpecGen: Automated Generation of Formal Program Specifications via
                  Large Language Models},
  booktitle    = {{ICSE}},
  pages        = {16--28},
  publisher    = {{IEEE}},
  year         = {2025},
  doi = {10.1109/ICSE55347.2025.00129}
}

@article{hoare-logic,
author = {Hoare, C. A. R.},
title = {An axiomatic basis for computer programming},
year = {1969},
issue_date = {Oct. 1969},
publisher = {Association for Computing Machinery},
address = {New York, NY, USA},
volume = {12},
number = {10},
issn = {0001-0782},
url = {https://doi.org/10.1145/363235.363259},
doi = {10.1145/363235.363259},
journal = {Commun. ACM},
month = oct,
pages = {576–580},
numpages = {5}
}

@inproceedings{contiki,
  author       = {Alexandre Peyrard and
                  Nikolai Kosmatov and
                  Simon Duquennoy and
                  Inria Lille and
                  Shahid Raza},
  title        = {Towards Formal Verification of {C}ontiki: {A}nalysis of the {AES-CCM}* Modules with {F}rama-{C}},
  booktitle    = {{EWSN}},
  pages        = {264--269},
  publisher    = {Junction Publishing, Canada/ {ACM}},
  year         = {2018},
  url = {https://dl.acm.org/doi/10.5555/3234847.3234910}
}

@article{preguss,
author = {Wang, Zhongyi and Lin, Tengjie and Chen, Mingshuai and Li, Haokun and Yang, Mingqi and Yi, Xiao and Qin, Shengchao and Luo, Yixing and Li, Xiaofeng and Gu, Bin and Lu, Liqiang and Yin, Jianwei},
title = {A Tale of 1001 LoC: Potential Runtime Error-Guided Specification Synthesis for Verifying Large-Scale Programs},
year = {2026},
issue_date = {April 2026},
publisher = {Association for Computing Machinery},
address = {New York, NY, USA},
volume = {10},
number = {OOPSLA1},
doi = {10.1145/3798268},
url = {https://doi.org/10.1145/3798268},
journal = {Proc. ACM Program. Lang.},
articleno = {160},
numpages = {29}
}

@inproceedings{VST-esop11,
  author       = {Andrew W. Appel},
  editor       = {Gilles Barthe},
  title        = {Verified Software Toolchain - (Invited Talk)},
  booktitle    = {Programming Languages and Systems - 20th European Symposium on Programming,
                  {ESOP} 2011, Held as Part of the Joint European Conferences on Theory
                  and Practice of Software, {ETAPS} 2011, Saarbr{\"{u}}cken, Germany,
                  March 26-April 3, 2011. Proceedings},
  series       = {Lecture Notes in Computer Science},
  volume       = {6602},
  pages        = {1--17},
  publisher    = {Springer},
  year         = {2011},
  url          = {https://doi.org/10.1007/978-3-642-19718-5\_1},
  doi          = {10.1007/978-3-642-19718-5\_1},
  bibsource    = {dblp computer science bibliography, https://dblp.org}
}

@article{sel4-tocs14,
  author       = {Gerwin Klein and
                  June Andronick and
                  Kevin Elphinstone and
                  Toby Murray and
                  Thomas Sewell and
                  Rafal Kolanski and
                  Gernot Heiser},
  title        = {Comprehensive Formal Verification of an {OS} Microkernel},
  journal      = {{ACM} Transactions on Computer Systems},
  volume       = {32},
  number       = {1},
  articleno    = {2},
  numpages     = {70},
  year         = {2014},
  doi          = {10.1145/2560537},
  url          = {https://doi.org/10.1145/2560537}
}

@inproceedings{sel4-sosp09,
  author       = {Gerwin Klein and
                  Kevin Elphinstone and
                  Gernot Heiser and
                  June Andronick and
                  David Cock and
                  Philip Derrin and
                  Dhammika Elkaduwe and
                  Kai Engelhardt and
                  Rafal Kolanski and
                  Michael Norrish and
                  Thomas Sewell and
                  Harvey Tuch and
                  Simon Winwood},
  title        = {{seL4}: Formal Verification of an {OS} Kernel},
  booktitle    = {Proceedings of the 22nd {ACM} Symposium on Operating Systems
                  Principles},
  pages        = {207--220},
  publisher    = {{ACM}},
  year         = {2009},
  doi          = {10.1145/1629575.1629596},
  url          = {https://doi.org/10.1145/1629575.1629596}
}

@inproceedings{RefinedC-pldi21,
  author       = {Michael Sammler and
                  Rodolphe Lepigre and
                  Robbert Krebbers and
                  Kayvan Memarian and
                  Derek Dreyer and
                  Deepak Garg},
  editor       = {Stephen N. Freund and
                  Eran Yahav},
  title        = {RefinedC: automating the foundational verification of {C} code with
                  refined ownership types},
  booktitle    = {{PLDI} '21: 42nd {ACM} {SIGPLAN} International Conference on Programming
                  Language Design and Implementation, Virtual Event, Canada, June 20-25,
                  2021},
  pages        = {158--174},
  publisher    = {{ACM}},
  year         = {2021},
  url          = {https://doi.org/10.1145/3453483.3454036},
  doi          = {10.1145/3453483.3454036},
  bibsource    = {dblp computer science bibliography, https://dblp.org}
}

@article{typestate-tse86,
  author       = {Robert E. Strom and
                  Shaula Yemini},
  title        = {Typestate: {A} Programming Language Concept for Enhancing Software
                  Reliability},
  journal      = {{IEEE} Trans. Software Eng.},
  volume       = {12},
  number       = {1},
  pages        = {157--171},
  year         = {1986},
  url          = {https://doi.org/10.1109/TSE.1986.6312929},
  doi          = {10.1109/TSE.1986.6312929},
  bibsource    = {dblp computer science bibliography, https://dblp.org}
}

@inproceedings{Vault-pldi01,
  author       = {Robert DeLine and
                  Manuel F{\"{a}}hndrich},
  editor       = {Michael Burke and
                  Mary Lou Soffa},
  title        = {Enforcing High-Level Protocols in Low-Level Software},
  booktitle    = {Proceedings of the 2001 {ACM} {SIGPLAN} Conference on Programming
                  Language Design and Implementation (PLDI), Snowbird, Utah, USA, June
                  20-22, 2001},
  pages        = {59--69},
  publisher    = {{ACM}},
  year         = {2001},
  url          = {https://doi.org/10.1145/378795.378811},
  doi          = {10.1145/378795.378811},
  bibsource    = {dblp computer science bibliography, https://dblp.org}
}

@inproceedings{permissive-interface-fse05,
  author       = {Thomas A. Henzinger and
                  Ranjit Jhala and
                  Rupak Majumdar},
  editor       = {Michel Wermelinger and
                  Harald C. Gall},
  title        = {Permissive interfaces},
  booktitle    = {Proceedings of the 10th European Software Engineering Conference held
                  jointly with 13th {ACM} {SIGSOFT} International Symposium on Foundations
                  of Software Engineering, 2005, Lisbon, Portugal, September 5-9, 2005},
  pages        = {31--40},
  publisher    = {{ACM}},
  year         = {2005},
  url          = {https://doi.org/10.1145/1081706.1081713},
  doi          = {10.1145/1081706.1081713},
  bibsource    = {dblp computer science bibliography, https://dblp.org}
}

@inproceedings{Java-interface-popl05,
  author       = {Rajeev Alur and
                  Pavol Cern{\'{y}} and
                  P. Madhusudan and
                  Wonhong Nam},
  editor       = {Jens Palsberg and
                  Mart{\'{\i}}n Abadi},
  title        = {Synthesis of interface specifications for Java classes},
  booktitle    = {Proceedings of the 32nd {ACM} {SIGPLAN-SIGACT} Symposium on Principles
                  of Programming Languages, {POPL} 2005, Long Beach, California, USA,
                  January 12-14, 2005},
  pages        = {98--109},
  publisher    = {{ACM}},
  year         = {2005},
  url          = {https://doi.org/10.1145/1040305.1040314},
  doi          = {10.1145/1040305.1040314},
  bibsource    = {dblp computer science bibliography, https://dblp.org}
}

@article{compcert-cacm09,
  author       = {Xavier Leroy},
  title        = {Formal verification of a realistic compiler},
  journal      = {Communications of the {ACM}},
  volume       = {52},
  number       = {7},
  pages        = {107--115},
  year         = {2009},
  doi          = {10.1145/1538788.1538814}
}

@inproceedings{VSU-esop21,
  author       = {Lennart Beringer},
  editor       = {Nobuko Yoshida},
  title        = {Verified Software Units},
  booktitle    = {Programming Languages and Systems - 30th European Symposium on Programming,
                  {ESOP} 2021, Held as Part of the European Joint Conferences on Theory
                  and Practice of Software, {ETAPS} 2021, Luxembourg City, Luxembourg,
                  March 27 - April 1, 2021, Proceedings},
  series       = {Lecture Notes in Computer Science},
  volume       = {12648},
  pages        = {118--147},
  publisher    = {Springer},
  year         = {2021},
  doi          = {10.1007/978-3-030-72019-3_5}
}

@misc{iso-15026,
  author       = {{ISO/IEC}},
  title        = {{ISO/IEC} 15026-1:2019 Systems and Software Engineering - Systems and
                  Software Assurance - Part 1: Concepts and Vocabulary},
  year         = {2019},
  publisher    = {{ISO/IEC}}
}

@incollection{assurance-case,
title = {Confidence as a product},
booktitle = {System Assurance},
publisher = {Morgan Kaufmann},
address = {Boston},
pages = {23-47},
year = {2011},
series = {The MK/OMG Press},
isbn = {978-0-12-381414-2},
doi = {https://doi.org/10.1016/B978-0-12-381414-2.00002-6},
url = {https://www.sciencedirect.com/science/article/pii/B9780123814142000026},
author = {Nikolai Mansourov and Djenana Campara}
}

@article{coqhammer-jar18,
  author       = {{\L}ukasz Czajka and Cezary Kaliszyk},
  title        = {Hammer for Coq: Automation for Dependent Type Theory},
  journal      = {J. Autom. Reason.},
  volume       = {61},
  number       = {1-4},
  pages        = {423--453},
  year         = {2018},
  doi          = {10.1007/s10817-018-9458-4}
}

@inproceedings{baldur-fse23,
  author       = {Emily First and Markus N. Rabe and Talia Ringer and Yuriy Brun},
  title        = {Baldur: Whole-Proof Generation and Repair with Large Language
                  Models},
  booktitle    = {Proceedings of the 31st ACM Joint European Software Engineering
                  Conference and Symposium on the Foundations of Software Engineering},
  pages        = {1229--1241},
  publisher    = {{ACM}},
  year         = {2023},
  doi          = {10.1145/3611643.3616243}
}

@inproceedings{proverbot9001,
  author       = {Alex Sanchez-Stern and Yousef Alhessi and Lawrence Saul and
                  Sorin Lerner},
  title        = {Generating Correctness Proofs with Neural Networks},
  booktitle    = {Proceedings of the 4th ACM SIGPLAN International Workshop on
                  Machine Learning and Programming Languages ({MAPL} 2020)},
  pages        = {1--10},
  publisher    = {{ACM}},
  year         = {2020},
  doi          = {10.1145/3394450.3397466}
}

@inproceedings{coqpilot-ase24,
  author       = {Andrei Kozyrev and Gleb Solovev and Nikita Khramov and Anton
                  Podkopaev},
  title        = {CoqPilot, a plugin for LLM-based generation of proofs},
  booktitle    = {Proceedings of the 39th IEEE/ACM International Conference on
                  Automated Software Engineering ({ASE} 2024)},
  pages        = {2382--2385},
  publisher    = {{ACM}},
  year         = {2024},
  doi          = {10.1145/3691620.3695357}
}

@inproceedings{palm-ase24,
  author       = {Minghai Lu and Benjamin Delaware and Tianyi Zhang},
  title        = {Proof Automation with Large Language Models},
  booktitle    = {Proceedings of the 39th IEEE/ACM International Conference on
                  Automated Software Engineering ({ASE} 2024)},
  pages        = {1509--1520},
  publisher    = {{ACM}},
  year         = {2024},
  doi          = {10.1145/3691620.3695521}
}

@inproceedings{cobblestone-icse26,
  author       = {Saketh Ram Kasibatla and Arpan Agarwal and Yuriy Brun and Sorin
                  Lerner and Talia Ringer and Emily First},
  title        = {Cobblestone: A Divide-and-Conquer Approach for Automating Formal
                  Verification},
  booktitle    = {Proceedings of the 48th International Conference on Software
                  Engineering ({ICSE} 2026)},
  pages        = {704--716},
  year         = {2026},
  doi          = {10.1145/3744916.3773178},
  eprint       = {2410.19940},
  eprinttype   = {arxiv}
}

@article{autorocq,
  author       = {Haoxin Tu and Huan Zhao and Yahui Song and Mehtab Zafar and Ruijie
                  Meng and Abhik Roychoudhury},
  title        = {Agentic Verification of Software Systems},
  journal      = {Proc. {ACM} Softw. Eng.},
  volume       = {3},
  number       = {FSE},
  pages        = {3558--3581},
  year         = {2026},
  doi          = {10.1145/3808164},
  eprint       = {2511.17330},
  eprinttype   = {arxiv}
}

@inproceedings{verus-sosp24,
  author       = {Andrea Lattuada and Travis Hance and Jay Bosamiya and Matthias Brun
                  and Chanhee Cho and Hayley LeBlanc and Pranav Srinivasan and Reto
                  Achermann and Tej Chajed and Chris Hawblitzel and Jon Howell and
                  Jacob R. Lorch and Oded Padon and Bryan Parno},
  title        = {Verus: A Practical Foundation for Systems Verification},
  booktitle    = {Proceedings of the 30th Symposium on Operating Systems Principles
                  ({SOSP} 2024)},
  pages        = {438--454},
  publisher    = {{ACM}},
  year         = {2024},
  doi          = {10.1145/3694715.3695952}
}

@article{laurel-oopsla25,
  author       = {Eric Mugnier and Emmanuel Anaya Gonzalez and Nadia Polikarpova and
                  Ranjit Jhala and Yuanyuan Zhou},
  title        = {Laurel: Unblocking Automated Verification with Large Language
                  Models},
  journal      = {Proc. {ACM} Program. Lang.},
  volume       = {9},
  number       = {{OOPSLA}},
  pages        = {1519--1545},
  year         = {2025},
  doi          = {10.1145/3720499}
}

@article{autoverus-oopsla25,
  author       = {Chenyuan Yang and Xuheng Li and Md Rakib Hossain Misu and Jianan Yao
                  and Weidong Cui and Yeyun Gong and Chris Hawblitzel and Shuvendu
                  Lahiri and Jacob R. Lorch and Shuai Lu and Fan Yang and Ziqiao Zhou
                  and Shan Lu},
  title        = {AutoVerus: Automated Proof Generation for Rust Code},
  journal      = {Proc. {ACM} Program. Lang.},
  volume       = {9},
  number       = {{OOPSLA}},
  pages        = {3454--3482},
  year         = {2025},
  doi          = {10.1145/3763174}
}

@inproceedings{slvc-icml26,
  author       = {Hanyang Wang and
                  Xiwei Wu and
                  Qinxiang Cao},
  title        = {{SL-VC}: A Benchmark and Automated Framework for Separation Logic
                  Verification Condition Proving},
  booktitle    = {Proceedings of the 43rd International Conference on Machine Learning
                  ({ICML} 2026)},
  series       = {Proceedings of Machine Learning Research},
  volume       = {306},
  address      = {Seoul, South Korea},
  publisher    = {{PMLR}},
  year         = {2026},
  note         = {Equal contribution}
}

@inproceedings{rango-icse25,
  author       = {Kyle Thompson and
                  Nuno Saavedra and
                  Pedro Carrott and
                  Kevin Fisher and
                  Alex Sanchez-Stern and
                  Yuriy Brun and
                  Jo{\~a}o F. Ferreira and
                  Sorin Lerner and
                  Emily First},
  title        = {Rango: Adaptive Retrieval-Augmented Proving for Automated Software
                  Verification},
  booktitle    = {{ICSE}},
  pages        = {347--359},
  publisher    = {{IEEE}},
  year         = {2025},
  doi          = {10.1109/ICSE55347.2025.00161}
}

@inproceedings{synver-ase25,
  author       = {Prasita Mukherjee and
                  Minghai Lu and
                  Benjamin Delaware},
  title        = {{LLM}-Assisted Synthesis of High-Assurance {C} Programs},
  booktitle    = {{ASE}},
  pages        = {3108},
  publisher    = {{IEEE}},
  year         = {2025},
  doi          = {10.1109/ASE63991.2025.00255}
}

@misc{fm-agent,
  author       = {Haoran Ding and
                  Zhaoguo Wang and
                  Haibo Chen},
  title        = {{FM-Agent}: Scaling formal methods to large systems via {LLM}-based {H}oare-style reasoning},
  year         = {2026},
  eprint       = {2604.11556},
  archiveprefix = {arXiv},
  primaryclass = {cs.SE}
}

@misc{revision-commit,
  author       = {{openHiTLS}},
  title        = {Commit ea0dc025: Harden {SLH-DSA} and {PQC} {PKI} Support},
  year         = {2026},
  howpublished = {\url{https://github.com/openHiTLS/openHiTLS/commit/ea0dc025635267ff79996994be7e1b7fc90946a2}},
  note         = {Upstream commit, August 5, 2026}
}

@online{opencode-native-goal-release,
  author = {{Anomaly}},
  title = {{OpenCode v2.2.1 Release}},
  year = {2026},
  url = {https://github.com/anomalyco/opencode/releases/tag/v2.2.1},
  urldate = {2026-11-02}
}

@misc{verdict,
  author = {Tengjie Lin and Zhongyi Wang and Haokun Li and Jianwei Yin and Mingshuai Chen},
  title = {{Verdict}: Multi-Agent Parallel Modular Verification},
  year = {2026},
  eprint = {2610.29147},
  archivePrefix = {arXiv},
  primaryClass = {cs.SE},
  doi = {10.48550/arXiv.2610.29147},
  url = {https://arxiv.org/abs/2610.29147},
  note = {To appear at ICSE 2027}
}

@online{openhitls-source,
 author = {{openHiTLS Project}},
 title = {{openHiTLS Source Repository}},
 url = {https://github.com/openHiTLS/openHiTLS},
 urldate = {2026-11-11},
 year = {2026}
}

@online{contiki-source,
 author = {{Contiki Project}},
 title = {{Contiki Source Repository}},
 url = {https://github.com/contiki-os/contiki},
 urldate = {2026-11-11},
 year = {2026}
}

@article{mcp-server-fsm,
  author = {Yutao Sun and Haokun Li and Mingshuai Chen and Guanyan Li and Zhongyi Wang and Shengchao Qin and Jianwei Yin},
  title = {Keeping {LLM} Agents on Procedure: Executable Finite-State Control for Long-Horizon Instruction Fidelity},
  journal = {arXiv preprint arXiv:2611.10884},
  year = {2026},
  eprint = {2611.10884},
  archivePrefix = {arXiv},
  primaryClass = {cs.AI},
  doi = {10.48550/arXiv.2611.10884},
  url = {https://arxiv.org/abs/2611.10884}
}

@misc{mldsa-apple,
  author = {{Apple Security Engineering and Architecture (SEAR) and Hardware Technologies Formal Verification}},
  title = {A Blueprint for Formal Verification of Apple corecrypto},
  year = {2026},
  month = may,
  url = {https://security.apple.com/blog/formal-verification-corecrypto/},
  note = {Published May 22, 2026; accessed November 12, 2026}
}

@article{quiver-pldi24,
 author={Simon Spies and Lennard G{\"a}her and Michael Sammler and Derek Dreyer},
 title={{Quiver}: Guided Abductive Inference of Separation Logic Specifications in {Coq}},
 journal={Proceedings of the ACM on Programming Languages},
 volume={8}, number={PLDI}, articleno={183}, year={2024},
 doi={10.1145/3656413}
}

@article{hats-pldi24,
 author={Zhe Zhou and Qianchuan Ye and Benjamin Delaware and Suresh Jagannathan},
 title={A {HAT} Trick: Automatically Verifying Representation Invariants Using Symbolic Finite Automata},
 journal={Proceedings of the ACM on Programming Languages},
 volume={8}, number={PLDI}, articleno={203}, year={2024},
 doi={10.1145/3656433}
}

@inproceedings{veristruct-tacas26,
 author={Chuyue Sun and Yican Sun and Daneshvar Amrollahi and Ethan Zhang and Shuvendu Lahiri and Shan Lu and David Dill and Clark Barrett},
 title={{VeriStruct}: {AI}-assisted Automated Verification of Data-Structure Modules in {Verus}},
 booktitle={Tools and Algorithms for the Construction and Analysis of Systems},
 year={2026}, doi={10.1007/978-3-032-22749-2_6}
}

@article{dualis-pldi26,
 author={Abhishek Uppar and Omar Muhammad and Sumanth Prabhu S and Deepak D'Souza and P. Madhusudan and Adithya Murali},
 title={Verification Modulo Tested Library Contracts},
 journal={Proceedings of the ACM on Programming Languages},
 volume={10}, number={PLDI}, year={2026}, doi={10.1145/3808305}
}

@inproceedings{schellhorn-rbtree-vstte22,
  author = {Gerhard Schellhorn and Stefan Bodenm{\"u}ller and Martin Bitterlich and Wolfgang Reif},
  title = {Separating Separation Logic -- Modular Verification of Red-Black Trees},
  booktitle = {Verified Software. Theories, Tools and Experiments (VSTTE 2022)},
  year = {2023},
  volume = {13800},
  series = {Lecture Notes in Computer Science},
  pages = {129--147},
  publisher = {Springer},
  doi = {10.1007/978-3-031-25803-9_8},
  url = {https://opus.bibliothek.uni-augsburg.de/opus4/frontdoor/index/index/docId/101817}
}

@inproceedings{dross-rbtree-nfm17,
  author = {Claire Dross and Yannick Moy},
  title = {Auto-Active Proof of Red-Black Trees in {SPARK}},
  booktitle = {NASA Formal Methods},
  year = {2017},
  publisher = {Springer},
  url = {https://www.adacore.com/uploads/blog/Auto-Active-Proof-of-Red-Black-Trees-in-SPARK.pdf}
}

@techreport{fips204,
  author       = {{National Institute of Standards and Technology}},
  title        = {Module-Lattice-Based Digital Signature Standard},
  institution  = {U.S. Department of Commerce},
  type         = {Federal Information Processing Standards Publication},
  number       = {FIPS 204},
  year         = {2024},
  month        = aug,
  doi          = {10.6028/NIST.FIPS.204}
}

@misc{alphaproof-code,
  author       = {Swarat Chaudhuri and Andrew Ferraiuolo and Anton Kovsharov and Tej Chajed and Sam Schlesinger and Karn Seth and Jonathan Katz and Charles Sutton and Pushmeet Kohli},
  title        = {{AlphaProof Code}: Verifying {Signal}'s Security at the Machine-Code Level},
  howpublished = {Google DeepMind},
  year         = {2026},
  month        = nov,
  url          = {https://google-deepmind.github.io/code-verification/}
}

@misc{anthropic-formal-verification,
  author       = {Keane Lucas and Aengus Lynch and Luna Tong and Luke Champine and Zygi Straznickas and Dino Distefano and Harshikaa Agrawal and Evyatar {Ben Asher} and Ava Howell and Sami Jawhar and Chris Rohlf and Nickolai Zeldovich and Nicholas Carlini and Rajashree Agrawal and Jason Gross},
  title        = {{LLMs} Are Making Formal Software Verification Feasible at a Much Larger Scale},
  howpublished = {Anthropic Frontier Red Team, in collaboration with Theorem},
  year         = {2026},
  month        = oct,
  url          = {https://www.anthropic.com/research/formal-verification}
}

\newpage
\appendix
\crefalias{section}{appendix}
\section*{Appendix}\label{sec:appendix}
\section{Scope of Runtime-Safety Guarantees}\label{app:vst-safety}

The Clight/VST verification model covers core classes of memory and arithmetic errors in C, including invalid memory reads and writes, signed-integer addition, subtraction, and multiplication overflow, and integer division by zero~\cite{compcert-cacm09,VST-esop11}.
Memory-access safety follows from the resource and permission requirements of VST's rules; arithmetic checks impose the corresponding value constraints.
Invalid deallocation is excluded under allocator contracts that require valid allocation provenance and ownership, while other external calls remain conditional on their specified behavior.
The Clight model differs from source-level ISO C~\cite{c-standard} in four respects.

\emph{Block-scoped lifetimes.}
Clight allocates addressable locals at function entry and releases them on return.
An access through a pointer that has outlived its block, as in \texttt{int *p; \{ int x = 7; p = \&x; \} return *p;}, is therefore permitted by the model.

\emph{Pointer formation.}
ISO C restricts pointer arithmetic to an array object and its one-past position; Clight places no condition on forming a pointer.
VST requires validity where a pointer is used, in dereferences and in comparisons.
Because functions are verified one at a time, and their specifications normally require pointer arguments to be valid, few computed pointers escape these requirements.

\emph{Effective types.}
Clight memory is untyped, so ISO C's effective-type (strict-aliasing) rules are not enforced.
The guarantee therefore assumes compilation without type-based alias analysis, as with \texttt{-fno-strict-aliasing}, the default in openHiTLS and the Linux kernel.

\emph{Evaluation order.}
The translation fixes one evaluation order, and the proofs cover only that order.

The claims remain conditional on the entry preconditions, external and allocator contracts, and source-to-Clight correspondence discussed in \cref{sec:assurance-case}.

\section{State Interpretation in a SHA-256 Context}\label{app:sha256-state-interpretation}

This case study illustrates the pure and spatial components of the state interpretation in \cref{def:interface-protocol} using the SHA-256 context in our evaluation subject.
We present a simplified interpretation focused on context fields; it is not a complete protocol or an additional verification result.
The context contains eight hash words, a 64-byte block buffer, message-length counters, and the fields \texttt{blocklen}, \texttt{outlen}, and \texttt{errorCode}.
The implementation uses \texttt{blocklen} to locate the next buffer position and \texttt{outlen} to determine how many hash words finalization emits.

\paragraph{Values and their representation.}
Let $c$ be a logical record of the context contents, and write
$n=c.\mathsf{blocklen}$, $d=c.\mathsf{outlen}$, and $e=c.\mathsf{errorCode}$.
Let $\mathsf{CtxMem}(p,c)$ describe ownership of the context storage at $p$ with the contents recorded in $c$.
For this illustration, numerical restrictions are displayed separately from this storage predicate.
A candidate pure condition is
\[
\begin{aligned}
  \mathsf{CtxWF}(c)\ \defeq\;&
    0\leq n<64\ \land\ 0\leq d\leq32\\
    &{}\land\ e\in\{0,\mathsf{Overflow}\},
\end{aligned}
\]
where $\mathsf{Overflow}$ denotes the implementation's input-overflow error code.
The bound on $n$ constrains buffer indexing; the bound on $d$ bounds the number of hash words read during finalization.
The restriction on $e$ describes the modeled error states, rather than an array bound.
These conditions do not supply storage or ownership: $\mathsf{CtxMem}(p,c)$ is needed to connect the logical record to memory.

For a state whose library-object map $\ell$ records such contexts, the context-related parts of its interpretation can be written as
\[
\begin{aligned}
  \llbracket\sigma\rrbracket_{\mathsf p}
    &\defeq \bigwedge_{p\mapsto c\in\ell}\mathsf{CtxWF}(c),\\
  \mathsf{LibObjRep}(\ell)
    &\defeq \mathop{\mbox{\large $*$}}\limits_{p\mapsto c\in\ell}
       \mathsf{CtxMem}(p,c).
\end{aligned}
\]
The global and borrowed-resource components are omitted here.
The pure component constrains each recorded context; the spatial component represents its fields and separates its storage from that of the other tracked contexts.

\paragraph{Changing values while preserving constraints.}
Suppose an initialized context has $n=10$ and a successful update appends five bytes without filling the block.
The implementation copies those bytes into the buffer and sets $n'=15$.
The successor record $c'$ differs from $c$, and $\mathsf{CtxMem}(p,c')$ describes the updated fields and buffer contents.
The pure obligations also refer to different values: $0\leq10<64$ before the call and $0\leq15<64$ afterward.
Preservation means establishing $\mathsf{CtxWF}(c')$ for the successor record; it does not require the record or its interpretation to remain unchanged.
This example assumes the call's remaining preconditions, including readable input bytes and the required non-overlap conditions.

Initialization must establish these constraints, and subsequent calls must re-establish them for their successor states.

\paragraph{Relation to the interface precondition.}
The update implementation consumes the bound on \texttt{blocklen} on paths that append input to the buffer.
An interface may require that bound directly, while a protocol edge supplies it from the source state's interpretation.
The edge must also supply the input-buffer resources and other conditions required by the interface.
Subsumption checks that these components together suffice for the interface precondition; the state constraint alone does not supply the input memory.

\paragraph{Comparison with the linked-list example.}
In \cref{ex:linked-list-components}, the pure state component is $\mathsf{True}$, while the spatial interpretation records each list's contents through $\mathsf{list\_rep}$.
The integer bound for a length query appears in that function's interface precondition and guard.
The SHA-256 interpretation above instead exposes numerical constraints shared by the modeled context states.
Both choices use the same definition of state interpretation. Appendix~\ref{app:mceliece-protocol} extends the illustration to resources with different allocation and initialization guarantees and to failure states that constrain subsequent calls.

\section{Functional Correctness Cases}
\label{app:functional-cases}

\subsection{Red--Black Tree}
\label{app:rbtree-functional}

\paragraph{Representation and invariants.}
The verified implementation contains nine functions, including five public interfaces: \texttt{tree\_new}, \texttt{rb\_insert}, \texttt{rb\_search}, \texttt{rb\_free}, and \texttt{tree\_free}.
An address-labeled logical tree relates abstract keys to concrete C nodes, parent pointers, and separately owned heap cells.
Strict search ordering requires every key in a node's left subtree to be smaller than its key and every key in its right subtree to be larger.
This structural invariant is distinct from the search operation's return-value contract.
The red--black invariants require legal node colors, a black root, no adjacent red nodes, and equal black height along root-to-leaf paths.
Rotations and recoloring pass through intermediate states that need not satisfy all final invariants; the representations and helper contracts must describe these states while tracking node identities, parent links, and disjoint resources.
Connecting these intermediate states to the public contracts is necessary both for the C body proofs and for subsequent interface calls.

\paragraph{Operation contracts.}
Successful creation establishes an empty valid tree; allocation failure returns null without new tree resources.
Insertion succeeds by adding the requested key and a fresh node while preserving strict search ordering and red--black balance, or fails while returning the original tree and its resources unchanged.
The successful-insertion contract separately tracks the key and address collections; it does not additionally assert preservation of every old address--key pairing.
Search returns null if and only if the key is absent, otherwise returns the address of a node carrying that key, and preserves the tree and its ownership.
Returning a node pointer does not transfer ownership to the caller.
Whole-tree destruction consumes a live tree's nodes and header; a null argument is a no-op.

\paragraph{Interface protocol.}
The protocol tracks multiple tree objects and distinguishes \texttt{Live} objects from \texttt{Poison} objects whose links may refer to freed nodes.
A successful creation introduces a \texttt{Live} tree, and insertion and search preserve that status and the required tree invariants.
The public subtree destructor \texttt{rb\_free} does not repair incoming parent or header links.
When it consumes a subtree owned by a live tree, the protocol changes the owning tree to \texttt{Poison}, records freed addresses, and retains ownership of the remaining nodes and header with their actual field values.
A poisoned tree cannot be used for insertion, search, or non-null whole-tree destruction; null no-ops and further destruction of complete surviving subtrees remain permitted under their edge conditions.
Other independently owned live trees remain usable.
Thus the protocol preserves the conditions needed by subsequent functional operations without pretending that partial destruction leaves a valid searchable tree.
It does not guarantee that every poisoned object's remaining header can eventually be reclaimed through the public APIs.

\paragraph{Precise composition.}
For an allocator model satisfying its explicit assumptions, the lifecycle theorem quantifies over finite abstract runs accepted by the transition relation.
For every step, it establishes \texttt{funspec\_sub} from the canonical API contract to the corresponding protocol-edge contract.
Five single-step proofs and induction over accepted runs establish this result; actual C calls additionally require the body proofs, argument guards, owned resources, and external premises.
The pure transition relation alone does not encode all return-branch conditions or constitute an execution theorem for an arbitrary concrete client.
The six-step witness creates one tree, inserts and searches for a key, frees the returned subtree and leaves its owner poisoned, then creates and destroys another tree.
It covers all five APIs with satisfiable step assertions, rather than executing a concrete C client or establishing that all resources are freed at the end.

\paragraph{Coverage and related work.}
The verified properties are strict search ordering and all the listed color and black-height invariants, together with the functional and resource contracts for this implementation's creation, insertion, search, and destruction operations over the C \texttt{int} key domain.
The four internal helpers---node allocation, left and right rotation, and insertion fixup---are verified in support of these public operations.
The implementation provides no single-key deletion operation; the verification does not cover such deletion, complexity bounds, termination, concurrency, or safety of arbitrary client-side field access.
Two unused static checker functions are absent from the final Clight and are not counted among the nine verified functions.

Existing work already establishes mechanized correctness of pointer-based red--black trees.
Schellhorn et al.~\cite{schellhorn-rbtree-vstte22} separate algebraic reasoning from pointer reasoning through refinement and separation logic, obtaining an implementation with parent pointers and iterative control flow.
Dross and Moy~\cite{dross-rbtree-nfm17} demonstrate auto-active functional verification of an array-based imperative red--black-tree implementation in SPARK; their reported development omits the equal-black-height proof and deletion.
SynVer~\cite{synver-ase25} uses LLMs and VST to synthesize and verify C programs from specifications, addressing program generation rather than the verification of a given C red--black-tree implementation with an explicit lifecycle protocol.

The experimenter specified the interfaces, legal interactions, input domain, and external assumptions, and reviewed the resulting formalization.
AI agents constructed or adapted the representations, contracts, lemmas, body proofs, and protocol bridges.
The experiment comprised two runs: the later library-and-protocol run reused successful function proofs from the earlier run, and both runs are included in its reported costs.
Thus ``AI-driven'' does not mean that requirements were discovered without human input or that the second run reconstructed every proof from scratch.

\subsection{ML-DSA}
\label{app:mldsa-functional}

\paragraph{Functional contracts.}
The case verifies the openHiTLS ML-DSA core implementation for parameter set 65, covering the core and NTT modules.
For \texttt{KeyGenInternal}, a successful return binds the contents of both output key buffers to the mathematical key-generation result for the supplied seed.
For \texttt{SignInternal}, a successful return binds the reported output length and corresponding signature bytes to the mathematical signing result for the supplied private key, message input, and randomness.
The signing contract distinguishes ordinary message input from an already computed message representative: the latter uses a separate mathematical variant that consumes the representative directly rather than hashing it again.
For \texttt{VerifyInternal}, successful C verification implies acceptance by the mathematical verifier for the supplied key, message, and signature.
The standard-model equivalences concern the ML-DSA-65 internal algorithms; they do not establish correctness of every public wrapper or other parameter set.

\paragraph{Preconditions and external assumptions.}
These guarantees apply to inputs satisfying the refinement contracts, including the required parameter configuration, buffer contents and sizes, and ownership conditions.
Signing additionally assumes successful sampling and a successful mathematical signing computation before the implementation's nonce counter wraps around.
Concretely, its precondition supplies a successful mathematical computation with a fuel bound satisfying $\mathit{fuel}\cdot l\leq 65535$, where $l=5$ for ML-DSA-65.
External SHAKE calls are governed by content-tracking contracts; their implementations are not verified by this case.
Key generation additionally assumes the functional contracts of the two bounded rejection samplers (\texttt{RejBoundedPolyEta2} and \texttt{RejBoundedPolyEta4}): on success, their output is the rejection sampling of the SHAKE-256 stream of the seed and is produced within a fixed bound.
These contracts are used as premises and are not proved.
The functional guarantees concern successful returns under these preconditions, not unconditional successful completion for arbitrary inputs.

\paragraph{Supporting proofs.}
Specification extraction produced 11 Rocq files and 1,662 physical lines,
covering the 49 FIPS~204 algorithms, the relevant FIPS~202 definitions, and an ML-DSA-65 instance.
The agent constructed this model from the text of FIPS~204, and experts reviewed the entire model against the standard.
Plantard, rounding, decomposition, and NTT results supplied the arithmetic and representation bridges.
Existing ML-DSA safety proofs provided reference body proofs for NTT, inverse NTT, and the vector wrapper.

\paragraph{Scale and configuration.}
The verification was developed alongside an early version of \ours{}, with successive contract and proof iterations informing framework improvements.
The two original C files contain 1,502 physical lines (core: 1,377;
NTT: 125). The verified preprocessed inputs contain
8,063 lines and generate 28,781 lines of Clight.
There are 55 function definitions and corresponding body proofs, counting
the two translation-unit-local Plantard copies separately. The 142,152 lines of delivered Rocq
exclude generated Clight and include reused safety proofs, specifications,
mathematical bridges, and supporting definitions.
The core run used an early \ours{} version with Claude Code as the agent
host and was developed with Rocq~9.1.1; the delivered proofs also check with Rocq~9.0.0, the version used for the other cases.

\paragraph{Two layers of functional evidence.}
The implementation-level refinement contracts and their body proofs bind concrete key and signature buffers to implementation-oriented mathematical functions.
Unlike the memory-safety proofs, these body proofs are not linked into a single theorem: each function is proved against the functional contracts of its callees, and the contracts of the forward and inverse NTT routines, which reside in the other translation unit, are used as premises.
The separate equivalence theorems relate those functions to the standard-derived model; the mathematical equivalences alone would not establish the C implementation's behavior.
Verification has a one-way success guarantee: a successful C return implies mathematical acceptance, but mathematical acceptance need not imply a successful C return, for example when allocation fails.
These claims do not establish signing termination or cryptographic security.

\section{An Interface Protocol for Classic McEliece: Resource States and Failure Recovery}
\label{app:mceliece-protocol}

The archived Classic McEliece verification provides a concrete instance of the interface-protocol model in \cref{sec:central-protocol}.
The development covers five translation units, 138 function definitions, and 15 public interfaces under disclosed contracts and external premises.
We focus on how its actual state definitions distinguish ownership from content guarantees, and how failure states affect later calls.
The model is not a claim of complete cryptographic correctness or unrestricted use of the source library.

\paragraph{State space.}
Let $\mathbb F$ contain the 15 archived public functions; below we abbreviate their names by \texttt{Gen}, \texttt{Cmp}, and so on.
Writing $\mathsf{Param}$ for parameter records and $\mathsf{Share}$ for shares, the object-state type is
\[
\begin{aligned}
s::={}&\mathsf{Fresh}
\mid\mathsf{Para}(z,h)\mid\mathsf{Pub}(z,h)\mid\mathsf{Prv}(z,h)\\
&\mid\mathsf{PubPrv}(z,h)\mid\mathsf{PubPrvW}(z,h)
\mid\mathsf{PubR}(z,h)\mid\mathsf{PubPrvR}(z,h),
\end{aligned}
\]
where $z\in\mathsf{Param}$ and $h\in\mathsf{Share}$.
These abbreviate the actual \texttt{ST\_*} constructors summarized below.

\begin{center}
\small
\begin{tabular}{@{}lp{0.66\linewidth}@{}}
\toprule
Archived constructor & Resource interpretation \\
\midrule
\texttt{ST\_fresh} & An owned context with null parameter and key pointers. \\
\texttt{ST\_para} & A configured context with no keys and access to its parameter record. \\
\texttt{ST\_pub}, \texttt{ST\_prv} & A context holding the public or private key, respectively. \\
\texttt{ST\_pubprv} & A context holding both keys with their ordinary representation predicates. \\
\texttt{ST\_pubprvW} & Both keys remain allocated, but the private-key predicate provides weaker initialization guarantees. \\
\texttt{ST\_pubR}, \texttt{ST\_pubprvR} & Public-only or two-key states additionally carrying a ready public-key matrix representation. \\
\bottomrule
\end{tabular}
\end{center}

The global state is $\gamma=(T,J)$, where $T$ is a parameter table and $J$ records issued parameter entries.
The borrowed-resource type is empty, so its list can only be $[]$.
A node is $n=(p,d,s)$, recording a context address, its allocation size, and its object state.
Thus the components of \cref{eq:protocol-state-space} are
\[
\mathbb S_G=\mathsf{list}(\mathsf{Param})\times\mathsf{list}(\mathsf{Param}),
\qquad \mathbb S_B=\{[]\},
\qquad \mathbb S_L=\mathsf{list}(\mathsf{Val}\times\mathbb Z\times\mathsf{ST}).
\]
A protocol state is $\sigma=((T,J),[],\ell)$; the archived initial state is $((T_0,[]),[],[])$, where $T_0$ is the fixed table of 12 supported parameter records.

\paragraph{State interpretation.}
Fix the allocator model, global addresses, and representation shares, which we omit from the notation.
The archived \texttt{lib\_pure} condition can be displayed as
\[
\begin{aligned}
\mathsf{LibPure}(\sigma)\equiv{}&
 T=T_0\land |T|=12
 \land\forall z\in T.\;\mathsf{ParamWF}(z)\\
&{}\land\forall j<|T|.\;\mathsf{alg}(T[j])=2402+j\\
&{}\land\forall(p,d,s)\in\ell.\;
 d=32\land
 \bigl(s=\mathsf{Fresh}\ \lor\
   (\mathsf{param}(s)\in T_0\cap J\\
&\hspace{6em}{}\land\mathsf{Readable}(\mathsf{share}(s)))\bigr).
\end{aligned}
\]
Here $j$ ranges over natural-number indices, and membership in the lists is understood as membership in their entries.
$\mathsf{ParamWF}$ abbreviates the archived \texttt{params\_wf}: it constrains algorithm dimensions and buffer lengths, including the public-key matrix size.
The interpretation has the resource structure
\[
\mathsf{Iassert}(\sigma)
=
\langle\mathsf{LibPure}(\sigma)\rangle\land
\left(\mathsf{GlobalRep}(T,J)*
 \mathop{\mbox{\large $*$}}\limits_{n\in\ell}\mathsf{NodeRep}(n)\right).
\]
The brackets $\langle\cdot\rangle$ here denote the pure assertion embedded in the archived VST predicate.
$\mathsf{GlobalRep}$ combines the parameter-table representation, its parameter-sharing resources, and the remaining constant globals.
The node predicates also contain local field constraints; in the paired presentation of \cref{eq:state-interpretation}, these contribute to the pure component alongside $\mathsf{LibPure}$.
In particular, $\mathsf{LibPure}$ alone is not the full pure component of that presentation.

For the two-key states, the critical difference is visible by abbreviating the common context and parameter representation as $\mathsf{Base}$:
\[
\begin{aligned}
\mathsf{NodeRep}(p,32,\mathsf{PubPrvW}(z,h))
 &=\exists pp,pk,sk,lc.\;
   \mathsf{Base}(p,pp,pk,sk,lc,z,h)\\
 &\qquad *\mathsf{PubRep}(pk,z)*\mathsf{PrvRep}_{W}(sk,z),\\
\mathsf{NodeRep}(p,32,\mathsf{PubPrvR}(z,h))
 &=\exists pp,pk,sk,lc.\;
   \mathsf{Base}(p,pp,pk,sk,lc,z,h)\\
 &\qquad *\mathsf{PubRep}_{R}(pk,z)*\mathsf{PrvRep}_{V}(sk,z).
\end{aligned}
\]
$\mathsf{Base}$ owns the context storage and allocation token, records its parameter, public-key, private-key, and external-context pointers, and supplies the parameter representation at share $h$.
It also requires $pp$ to select the recorded parameter entry and $lc$ to be a pointer or null.
$\mathsf{PrvRep}_{W}$ and $\mathsf{PrvRep}_{V}$ abbreviate \texttt{CMPrivateKey\_repW} and \texttt{CMPrivateKey\_repValpha}; the latter provides stronger defined-value guarantees.
$\mathsf{PubRep}_{R}$ abbreviates \texttt{CMPublicKey\_ready\_rep}, strengthening the ordinary public-key representation.
These predicates distinguish content guarantees while retaining ownership; neither readiness nor ownership asserts cryptographic correctness.

\paragraph{A guard and its branches.}
For \texttt{Gen}, the archive selects the argument pointer through a separate target index $t$ and uses a unit logical input.
To match the state-independent argument map in \cref{eq:interface-schema}, we write $w=(t,p)$ and require $p$ to equal the selected node's pointer.
With zero-based indexing of $\ell$, its input mapping and guard are
\[
\begin{aligned}
\mathsf{args}_{\mathsf{Gen}}(t,p)&=(p),\\
G_{\mathsf{Gen}}(\sigma,(t,p))
 &\equiv \exists d,s.\;\ell[t]=(p,d,s)\land\mathsf{StateWF}(s),\\
\mathsf{StateWF}(\mathsf{Fresh})&=\mathsf{True},\\
\mathsf{StateWF}(s(z,h))&=\mathsf{ParamWF}(z)\land\mathsf{Readable}(h).
\end{aligned}
\]
The guard requires the selected node to exist and binds its pointer to the argument.
The last equation applies to each non-fresh constructor.
This merely makes the archived target-to-argument correspondence explicit in the paper notation.
The call additionally borrows the allocator resource, $\mathsf{BorrowIn}_{\mathsf{Gen}}(w)=\mathsf{Mgr}$.
Take $O_{\mathsf{Gen}}=\mathsf{Param}\times\mathsf{Share}\times\mathbb Z$, with $o=(z,h,r)$.
Let $\sigma[t\leftarrow s]$ replace the selected node's state by $s$, preserving its pointer and resetting the recorded context size to 32 as in the archive.
Then its branch list, displayed as a set, is
\[
\begin{aligned}
\mathsf{Br}_{\mathsf{Gen}}(\sigma,w,(z,h,r))=\{&
(\sigma[t\leftarrow s_r],\widehat r,\mathsf{True},\mathsf{Mgr}),\\
&(\sigma[t\leftarrow\mathsf{Para}(z,h)],\widehat r,\mathsf{True},\mathsf{Mgr}),\\
&(\sigma[t\leftarrow\mathsf{Fresh}],\widehat r,\mathsf{True},\mathsf{Mgr})\},
\end{aligned}
\]
Here
\[
s_r=
\begin{cases}
\mathsf{PubPrvR}(z,h),&\widehat r=0,\\
\mathsf{PubPrvW}(z,h),&\widehat r\ne0.
\end{cases}
\]
Here $\widehat r$ abbreviates the C integer value \texttt{Vint (Int.repr r)}, and zero denotes its zero value.
Each tuple uses the field order of \cref{eq:branch}.
The returned allocator resource lies outside the successor's tracked resources.
The formulas retain the archive's multiple candidate successors for a given return code; they do not assert that every listed branch is reachable.

One documented implementation path explains the need for the weaker W state.
After key-buffer allocation succeeds, random generation can fail before the private key's \texttt{controlbits} buffer is initialized.
The function returns an error while the context still owns the allocated keys.
The postcondition must therefore account for those resources without claiming the content guarantees of a successful generation.
Other failures may have made different progress; a nonzero return code does not by itself identify a unique successor.

\paragraph{Subsequent calls and recovery.}
Write $\mathsf{NotW}(s)$ when $s$ is not a $\mathsf{PubPrvW}$ state.
The archived guards for comparison and cleanup can be summarized using the generation guard:
\[
\begin{aligned}
G_{\mathsf{Cmp}}(\sigma,((t,p),(u,q)))
 &\equiv G_{\mathsf{Gen}}(\sigma,(t,p))\land G_{\mathsf{Gen}}(\sigma,(u,q))\\
 &\qquad
 \land\mathsf{NotW}(s_t)\land\mathsf{NotW}(s_u),\\
G_{\mathsf{FreeCtx}}(\sigma,(t,p))&\equiv G_{\mathsf{Gen}}(\sigma,(t,p)).
\end{aligned}
\]
Here $s_t,s_u$ are the states of the selected nodes; $\mathsf{NotW}$ is the archive's \texttt{pubprvW\_gate} on existing nodes.
Thus \texttt{Cmp} excludes W states, whereas \texttt{Gen} and \texttt{FreeCtx} permit another generation attempt or cleanup.
A retry does not automatically establish readiness: its selected successor depends on the new outcome, while cleanup removes the node.
These paths illustrate why failure need not mean either full rollback or loss of all ownership.

The restriction also exposes an important assurance boundary.
The archive documents a source-level path in which a random-generation failure leaves uninitialized \texttt{controlbits} that a subsequent comparison may consume.
The delivered proof excludes that path through the interface and protocol conditions; it does not repair the source defect or establish that the original library intended this restriction.
This is precisely a restriction whose justification matters to the conformance review of \cref{sec:cp-quality}.

\paragraph{Preserving facts across calls.}
The archive separately records a specification limitation: \texttt{Cmp}'s contract loses the stronger public-key readiness information, although this does not mean the implementation changes the bytes.
Consequently, a subsequent \texttt{Encaps} call requiring that information cannot in general be justified directly from the comparison postcondition.
The recorded witness includes \texttt{Gen} success followed by \texttt{Encaps} and then \texttt{Cmp}, rather than demonstrating the reverse order of the latter two calls.
An extra stay-put branch does not recover the missing universal guarantee.
This example connects the state interpretations and guards to per-edge subsumption: the interface postcondition must preserve the facts needed to establish the selected successor interpretation and support later calls.
It also shows why a checked protocol and a non-vacuity witness do not replace review of the supported call sequences.

\section{A Covering Witness for the Linked-List Protocol}
\label{app:covering-witness}

The five-interface library in \cref{fig:motivation-2-library} provides a small example of the checks in \cref{sec:cp-quality}.
We use its simplified successful behaviors: creation produces an empty list, insertion adds one element, removal deletes that element, and destruction releases the list.
This is an illustrative witness for the paper's example, not a new mechanized client proof.

\paragraph{Starting resources.}
Use the states of \cref{ex:linked-list-components}, with $\sigma_{\mathrm{init}}=(\mu_{\mathsf{mgr}},\emptyset,\emptyset)$.
The borrowed-resource and object components are empty, but the interpretation still requires $\mathsf{mem\_mgr}\;M\;gv$.
For this illustration, $\mathsf{Init}_L$ includes that allocator resource, which must be supplied by the allocator setup; an empty object map alone cannot provide it.
No allocated list is assumed at the starting state.

\paragraph{A finite covering run.}
Holding the global and borrowed components fixed, write only the object map $\ell$:
\[
\begin{aligned}
\emptyset
 \xrightarrow{\mathsf{new}} \{p\mapsto []\}
  \xrightarrow{\mathsf{push}(p,1)} \{p\mapsto [1]\}
  \xrightarrow{\mathsf{length}(p)} \{p\mapsto [1]\}
  \xrightarrow{\mathsf{remove}(p,1)} \{p\mapsto []\}
  \xrightarrow{\mathsf{free}(p)} \emptyset.
\end{aligned}
\]
The labels abbreviate the concrete function arguments; each full $w$ also records the current abstract list and any other logical inputs required by its contract.
The handle $p$ is selected by the successful creation branch, not supplied before creation.
Every interface function occurs once, and each displayed transition must satisfy its guard and have a branch with the displayed successor.
The allocation-success choice describes one path; coverage does not require enumerating allocation failures or all possible list contents.

\paragraph{Resource satisfiability.}
Validity of these transitions is checked separately from satisfiability of their resource conditions.
For example, the length call uses $w=(p,[1])$.
Its guard holds because the list is recorded in $\ell$ and its length fits in a C integer.
Its edge precondition requires the allocator resource and $\mathsf{list\_rep}\;p\;[1]$, with no additional borrowed resource.
A one-element list with the required ownership, compatible with the allocator resource, supplies a model of this assertion.
The selected postcondition has return value $1$ and preserves the same resources, so it too is satisfiable.
The other steps require corresponding models for the empty list, insertion, removal, and release; these must respect the same allocator model and the inputs and branches recorded in the run.

These checks rule out covering an interface solely through an impossible resource condition.
They do not themselves connect the per-step resource models by concrete executions of the function bodies.
Establishing that a client can execute this particular path would additionally require showing that the selected outcomes can occur successively.
The witness serves as a check of coverage and local satisfiability, while the soundness theorem separately applies to executions that realize valid runs.

\section{Detailed Analysis of Agent Baselines}\label{app:agent-baselines}

This appendix examines the specification, proof, and completion evidence for the two OpenCode baselines in \cref{sec:scaffolding-ladder}.

\subsection{\texorpdfstring{\textsc{OC}\textsubscript{bare}}{OC-bare}}\label{app:bare-agent-audit}

This subsection expands the \textsc{OC}\textsubscript{bare} result in \cref{sec:scaffolding-ladder} using the archived specifications, Clight programs, body proofs, and corrected audit findings.
We examine the memory-safety verification results for SHA-256/224; functional correctness of the hash computation is outside this analysis.
The development contains specifications for all 25 functions and body proofs for 11; the remaining 14 functions have no body proofs.
A clean rebuild succeeds for all six delivered Rocq source files, confirming that the available proof development is reproducible, but not that its contracts describe the intended library correctly or that verification is complete.

\paragraph{Specifications: an error-path counterexample.}
The two initialization specifications use \texttt{ctxFreeable} as their resource precondition and \texttt{ctxAny} as their resource postcondition, both applied to a share and a pointer.
The former allows a null pointer; the latter requires ownership of a complete context object at that pointer.
The following definitions are from the development, with formatting simplified:
\begin{verbatim}
Definition ctxAny (sh: share) (p: val) : mpred :=
  EX v: reptype t_ctx, data_at sh t_ctx v p.
Definition ctxFreeable (sh: share) (p: val) : mpred :=
  orp (!! (p = nullval))
      (orp (ctxAny sh p) (ctxAnyByte sh p)).
\end{verbatim}
The audit's successful probe, \texttt{probe3.v}, establishes both entailments below, and an independent recompilation confirms them:
\begin{verbatim}
emp |-- ctxFreeable sh nullval
ctxAny sh nullval |-- FF
\end{verbatim}
Thus, the resource precondition admits a null context, but the postcondition cannot hold at that pointer.
The implementation checks for NULL, calls the error-reporting routine, and returns an error code; whenever that routine returns normally, this execution contradicts the required postcondition.
This is a contract defect that requires correcting the error-path postcondition, rather than merely completing a proof script.
The normal-return qualification matters: an unsatisfiable postcondition alone does not rule out a diverging implementation under partial correctness.

\paragraph{Proofs: limited coverage and unresolved internal calls.}
Of the 11 proved bodies, four are multi-block stubs that return a constant error code, five wrap a single external call, and two forward to \texttt{Init} or \texttt{Init224}.
None contains a loop, a data-dependent branch, or a direct memory read or write; the compression, update, and finalization implementations remain unproved.
Consequently, the fraction 11/25 describes function coverage, not the fraction of proof effort completed.
The two initialization-wrapper theorems are valid modular judgments relative to the callee specifications in \texttt{Gprog}; they do not establish that the actual callees satisfy those specifications.
Their dependency remains unresolved because the initialization bodies have not been proved and their contracts contain the error-path mismatch above.
This does not make the wrapper theorems false or devoid of information on every input, but it prevents their use as a closed assurance argument for the real library.

\paragraph{Resource accounting: deallocation without provenance.}
The external \texttt{ClearFree} contract requires \texttt{ctxFreeable sh p} and returns an empty resource assertion.
Its non-null alternatives describe object storage without recording that the object was obtained from an allocator.
For example, ownership of a suitably typed stack object can satisfy such an object predicate, although passing that object to a conventional deallocator is invalid.
A sound deallocation interface therefore needs allocation provenance, expressed through an appropriate ownership token or an equivalent allocator model; the absence of a particular predicate name alone is not the defect.

Separately, the contract for \texttt{MdCommonGetParam} does not provide a validated model of accesses through its parameter pointer, and the actual provider was not checked.

\paragraph{Dependencies: admitted support and unvalidated source correspondence.}
The development's own files contain no explicit \texttt{Admitted}, but the body proofs depend on \texttt{convertPre\_helper3} in the installed VST legacy specification-conversion module.
That lemma ends with \texttt{Admitted}; its inclusion in the toolchain does not by itself justify accepting it as a trusted assumption.
This dependency is specific to the legacy conversion path used here, not an unavoidable property of every VST proof using \texttt{start\_function}.
Eliminating the dependency requires a proved replacement or a specification path that avoids it.

The agent supplied eight stub headers: seven replacements for missing headers and a supplementary \texttt{md\_common\_decl.h} declaring the external parameter helper, injected through the compiler's \texttt{-include} option.
Their fidelity to the intended upstream declarations was not established.
The audit did not replay the original-C-to-Clight translation or validate the external providers against their assumed contracts.
Therefore, the checked proof objects cannot yet support a conclusion about the intended source library under justified external assumptions.

\paragraph{Completion checks and scope of the evidence.}
The Makefile names ten handwritten modules, only four of which are present; two generated Clight modules supply additional dependencies.
It returns success despite the six missing handwritten modules, so a successful build checks the files that exist without establishing that all required proofs have been delivered.
The run also supplies no library-level composition theorem closing the internal implementation obligations under explicit external assumptions.
These findings establish an incomplete verification result; they do not establish that the target library is unsafe, that all generated contracts are wrong, or that another method would necessarily detect each defect.

\subsection{\texorpdfstring{\textsc{OC}\textsubscript{goal}}{OC-goal}}\label{app:oc-goal-audit}

\paragraph{Configuration and stopping behavior.}
This baseline uses OpenCode~2.2.1 with native \texttt{/goal} support~\cite{opencode-native-goal-release} and GLM-5.3-Flash.
The task requires all 25 function-body proofs, library composition, clean rebuilding, and disclosure of external assumptions, without \ours's verification-specific workflow.
Because native \texttt{/goal} support requires a different OpenCode version, the comparison with the bare agent does not isolate the effect of \texttt{/goal} alone.
The final observation follows a resumed run and supersedes the earlier interruption report.
The native \texttt{/goal} mechanism automatically paused execution after 30 minutes without activity while waiting for background work, at round 28 of 50, with the goal unmet and the token budget unexhausted.
This establishes non-completion in the observed run, not an upper bound on what further execution could achieve.

\paragraph{Proof coverage and composition.}
The last compiled milestone contains 19 of the 25 required body proofs.
The six remaining functions are \texttt{CompressBlock}, \texttt{Init224}, \texttt{Init}, \texttt{MBFreeCtx}, \texttt{Update}, and \texttt{Final}.
The final sources contain aborted or missing proofs for these obligations; they do not fill the gaps with admitted theorems.
\texttt{Compose.v} contains two library-hygiene theorems and a conjunction of the 19 available body proofs, which does not discharge the requested 25-body composition obligation.

\paragraph{Final-source consistency.}
The 19/25 count refers to a compiled milestone, not a confirmed clean rebuild of the final source tree.
Subsequent specification edits left \texttt{Verif\_sha2\_base.v} newer than its compiled object; later composition and audit objects depended on the earlier object.
Consequently, their presence does not certify the consistency of the final specifications and proofs.
The earlier interruption audit also found a duplicate \texttt{Update\_spec} definition that prevented rebuilding that snapshot; this historical defect is not asserted to persist in the resumed final snapshot.
These observations distinguish progress in local proofs from delivery of a consistent verification artifact.

\paragraph{Comparison with the delivered assurance case.}
\ours completed its declared verification obligations with human review of the contracts, assumptions, and protocol domain.
That domain remains explicit: its SHA GetParam lifecycle edges cover only the zero-sentinel case, rather than successful parameter queries; the per-call contracts have a broader domain.
Completion therefore refers to the reviewed obligations, not unrestricted coverage of every possible client call.

\section{Experimental Run Configurations}\label{app:run-configurations}

\Cref{tab:run-environments} summarizes the four experimental environments.
All RQ3 baseline and \ours runs use E2; their agent-host differences are specified in \cref{sec:scaffolding-ladder}.

\begin{table}[t]
\centering\small
\setlength{\tabcolsep}{4pt}
\begin{tabular}{@{}lp{.29\linewidth}p{.34\linewidth}p{.22\linewidth}@{}}
\toprule
Env. & Hardware / OS & Verification toolchain & Benchmarks \\
\midrule
E1 & AMD Threadripper PRO 5975WX, 64 cores, 1 TB RAM; Linux & Rocq 9.0.0 (ML-DSA: 9.1.1), VST 2.16, CompCert 3.17 & LL, RBT (safety), CTK, ML-DSA \\
\addlinespace
E2 & AMD EPYC 7542 VM, 128 vCPUs, 74 GB RAM; Ubuntu 24.04.2 & Rocq 9.0.0, VST 2.16*, CompCert 3.17 & SHA; all RQ3 runs \\
\addlinespace
E3 & Apple MacBook Air (M5, 24 GB RAM); macOS Tahoe 26.4 & Rocq 9.1.0 (first execution); Rocq 9.0.0, VST 2.16*, CompCert 3.17 (second execution) & RBT (functional) \\
\addlinespace
E4 & Intel Core i9-12900KF, 16 cores (24 threads), 63 GB RAM; Ubuntu 24.04.3 on WSL2 & Rocq 9.0.0, VST 2.16*, CompCert 3.17 & MCE \\
\bottomrule
\end{tabular}
\caption{Experimental environments and benchmark assignments. Benchmark abbreviations follow \cref{sec:evaluation-setup}.}
\label{tab:run-environments}\label{tab:run-configurations}
\end{table}

VST~2.16* denotes VST~2.16 extended with lemmas, proved in Rocq without additional axioms, for stores and loads that reinterpret the same bytes at a different type; it is included in the supplementary package.

\end{document}